\documentclass[a4paper, 10pt]{article}
\usepackage{authblk}

\usepackage{stmaryrd}
\usepackage[british]{babel}
\usepackage[utf8]{inputenc}

\usepackage{amsmath,amsthm,amssymb,amstext, amscd}
\usepackage{tikz-cd}
\usepackage{hyperref}
\usepackage{stmaryrd}
\usepackage{booktabs}
\usepackage{fouriernc}
\usepackage[T1]{fontenc}

\usepackage{todonotes}

\theoremstyle{plain}
\newtheorem{Theorem}{Theorem}
\newtheorem{Proposition}[Theorem]{Proposition}
\newtheorem{Corollary}[Theorem]{Corollary}
\newtheorem{Lemma}[Theorem]{Lemma}

\newtheorem{Algorithm}[Theorem]{Algorithm}

\theoremstyle{definition}
\newtheorem{Definition}[Theorem]{Definition}
\theoremstyle{remark}
\newtheorem{Remark}[Theorem]{Remark}

\usepackage{authblk}

\title{Representations and evaluation strategies for feasibly approximable functions\textsuperscript{\small1}}
\author {Michal Kone\v{c}n\'y \\ Aston University, Birmingham, UK\\ \texttt{\small m.konecny@aston.ac.uk}
\and 
Eike Neumann\textsuperscript{\small 2} \\ University of Oxford, UK \\ \texttt{\small eike.neumann@ox.ac.uk}}
\makeindex
\date{}

\newcommand{\ie}{\textit{i.e.}, ~}
\newcommand{\eg}{\textit{e.g.}, ~}

\newcommand{\B}{\mathcal{B}}
\newcommand{\ts}{2^*}

\newcommand{\R}{\mathbb{R}}

\newcommand{\Q}{\mathbb{Q}}
\newcommand{\D}{\mathbb{D}}
\newcommand{\N}{\mathbb{N}}

\newcommand{\M}{\mathcal{M}\!}
\newcommand{\qcb}{\operatorname{qcb}}

\renewcommand{\P}{\operatorname{P}}
\newcommand{\FP}{\operatorname{FP}}

\newcommand{\NP}{\operatorname{NP}}

\newcommand{\I}{{\mathcal{I}}}

\newcommand{\ID}{{\I\D}}
\newcommand{\CI}{C\left([-1,1]\right)}
\newcommand{\Cinf}{C^{\infty}[-1,1]}
\newcommand{\CT}{\operatorname{CT}}
\newcommand{\SCT}{\operatorname{SCT}}

\newcommand{\id}{\operatorname{id}}

\newcommand{\dom}{\operatorname{dom}}
\newcommand{\range}{\operatorname{range}}

\newcommand{\eval}{\operatorname{eval}}

\newcommand{\Op}{\operatorname{Op}}
\newcommand{\Const}{\operatorname{Const}}
\newcommand{\Fix}{\operatorname{Fix}}
\newcommand{\Free}{\operatorname{Free}}

\renewcommand{\div}{\operatorname{div}}

\newcommand{\primit}{\operatorname{primit}}
\newcommand{\paramax}{\operatorname{paramax}}

\newcommand{\const}{\operatorname{const}}

\newcommand{\Fun}{\operatorname{Fun}}
\newcommand{\BFun}{\operatorname{BFun}}
\newcommand{\DBFun}{\operatorname{DBFun}}
\newcommand{\Poly}{\operatorname{Poly}}
\newcommand{\PPoly}{\operatorname{PPoly}}
\newcommand{\Frac}{\operatorname{Frac}}
\newcommand{\PFrac}{\operatorname{PFrac}}
\newcommand{\PAff}{\operatorname{PAff}}

\newcommand{\LPoly}{\operatorname{LPoly}}
\newcommand{\LPPoly}{\operatorname{LPPoly}}
\newcommand{\LFrac}{\operatorname{LFrac}}

\newcommand{\Set}[2]{\left\{ #1 \; \mid \; #2 \right\}}
\newcommand{\norm}[1]{\left| #1 \right|}
\newcommand{\sem}[1]{\left\llbracket #1 \right\rrbracket}
\newcommand{\floor}[1]{\left\lfloor #1 \right\rfloor}

\newcommand{\Gev}{\operatorname{Gev}}

\newcommand{\name}[1]{{#1}}

\newcommand{\Prod}{\operatorname{Prod}_{\omega}}

\begin{document}

\maketitle

\begin{abstract}
A famous result due to Ko and Friedman (1982) asserts that
the problems of integration
and maximisation of a univariate real function
are computationally hard in a well-defined sense.
Yet, both functionals are routinely computed 
at great speed in practice.

We aim to resolve this apparent paradox by studying classes
of functions which can be feasibly integrated and maximised,
together with representations for these classes of functions
which encode the information which is necessary to
uniformly compute integral and maximum in polynomial time.
The theoretical framework for this is the second-order complexity
theory for operators in analysis
which was introduced by Kawamura and Cook (2012).

The representations we study are based on approximation
by polynomials, piecewise polynomials, and rational functions.
We compare these representations with respect to polytime reducibility.

We show that the representation based on approximation by piecewise polynomials
is polytime equivalent to 
the representation based on approximation by
rational functions. 

With this representation, all terms 
in a certain language, which is expressive
enough to contain the maximum and integral
of most functions of practical interest,
can be evaluated in polynomial time.
By contrast, both the representation based on
polynomial approximation
and the standard representation
based on function evaluation,
which implicitly underlies the Ko-Friedman result,
require exponential time to evaluate certain
terms in this language.

We confirm our theoretical results by
an implementation in Haskell,
which provides some evidence that second-order 
polynomial time computability
is similarly closely tied with 
practical feasibility
as its first-order counterpart.
\end{abstract}

\footnotetext[1]{\includegraphics[scale=0.04]{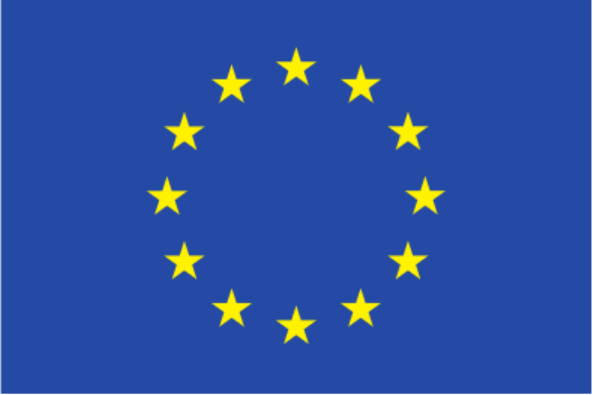} This project has received funding from the European Union’s Horizon 2020 research and innovation programme under the Marie Skłodowska-Curie grant agreement No 731143.}
\footnotetext[2]{Most of this work was carried out when Eike Neumann was affiliated with Aston University.}

\section{Introduction}

Consider the integration and maximisation functionals on the space \(\CI\) of univariate continuous functions over the compact interval $[-1,1]$:
\[
    f \mapsto \int_{-1}^1 f(x)\,\mathrm{d}x
    \qquad
    \text{and}
    \qquad
    f \mapsto \max_{x\in [-1,1]} f(x)
\]

Both functionals constitute fundamental basic operations in numerical mathematics.
They are considered to be easy to compute for functions that occur in practice.
It was hence surprising that when Ko and Friedman \cite{KoFriedman} introduced a 
rigorous formalisation of computational complexity in real analysis and analysed the computational complexity of these functionals
within this model, they found that both problems are computationally hard in a well-defined sense.
They constructed an infinitely differentiable polytime computable function $f_0 \colon [-1,1] \to \R$ such that the function
$g(x) = \int_{-1}^x f_0(t) \mathrm{d}t$ is again polytime computable if and only if $\FP = \sharp \P$
and
an infinitely differentiable polytime computable function $f_1 \colon [-1,1] \to \R$ such that the function
$h(x) = \max_{t \in [-1, x]} f_1(t)$ is again polytime computable if and only if $\P = \NP$. 
Moreover, the real number $g(1) = \int_{-1}^1 f_0(t) \mathrm{d}t$ is polytime computable if and only if $\FP_1 = \sharp \P_1$, and the number $h(1) = \max_{t \in [-1,1]} f_1(t)$ is again polytime computable if and only if $\P_1 = \NP_1$.

This obvious discrepancy between practical observations and theoretical predictions 
deserves further discussion.
We will focus on two possible explanations for this observation:

\begin{itemize}\label{explanations}
  \item \textbf{Accuracy of results.} 
  Hardness in the
  theoretical results refers to how hard it is to compute
  the values of the function to an arbitrary accuracy.
  An algorithm for computing a real number takes as input
  a natural number $n$, encoded in unary, and outputs
  an approximation to $x$ to $n$ bits of accuracy.
  An algorithm for computing a real function $f$ takes as
  input a real number $x$, encoded as an oracle which maps
  accuracy requirements to approximations, and a natural
  number $n$, encoded in unary, and is required to output
  an approximation to $f(x)$ to $n$ bits of accuracy.
	The running time of the algorithm is a function of $n$
	which measures the number of steps the algorithm takes.
  By contrast, practitioners usually work at a fixed floating-point
  precision, which implies a fixed maximum accuracy. 
  It hence may not be justified to measure the complexity in the
  output accuracy, and other complexity parameters should be 
  considered more important.
  In fact, if one relaxes the definition of polytime computability
  such that in both the definition of real number computation
  and real function computation the 
  requirement that the approximation be correct to $n$ bits
  of accuracy
  is relaxed to the requirement that the approximation 
  be $1/n$ close to the true value, then 
  the range and integral of every polytime computable function
  are polytime computable.
  So maybe the theoretical infeasibility of these functionals is 
  an artefact of poorly chosen normalisation.
  
  \item \textbf{Representation of functions.}  Theoreticians
  use a simple representation (which we call $\Fun$) that
  treats all continuous functions equally, in the sense that 
  a function is polynomial time computable if and only if it
  has a polynomial time computable $\Fun$-name.  
  Practitioners, on the other hand, 
  tend to work on a much more restricted class of functions.
	They tend to work with functions which are given
  symbolically or which can be approximated well by certain kinds of
  (piece-wise) polynomial or rational functions.  
  As not every polynomial time computable function can be approximated
  by polynomials or rational functions in polynomial time,
  the implicit underlying representations favour a certain class of 
  functions, for which it is easier to compute integral and range.
\end{itemize}

The aim of this paper is to discuss these different explanations
both from a theoretical and a practical perspective and to resolve
the apparent contradiction between the theoretical hardness results
and practical observations.
To this end we study the computational complexity of the maximisation and integration
functionals with respect to various representations of continuous real functions within
the uniform framework of second-order complexity theory, introduced by Kawamura and Cook \cite{KawamuraCook}, 
and compare the practical performance of algorithms which use these representations on a small family of benchmark problems.

\paragraph{Classes of feasibly approximable functions.}

The complexity of integration and maximisation of univariate real-valued functions has been studied by various authors:
M\"uller \cite{Mueller} showed that if $f$ is a polytime \emph{analytic} function, then the function $g(x) = \int_{-1}^x f(t) \mathrm{d}t$ is again polytime (and analytic), and the function $h(x) = \max_{t \in [-1, x]} f(t)$ is again polytime (but not differentiable in general). This result was generalised by Labhalla, Lombardi, and Moutai \cite{LabhallaEtAl} to the strictly larger class of polytime functions in \emph{Gevrey's hierarchy}, a class of infinitely differentiable functions whose derivatives satisfy certain growth conditions. These functions are characterised in \cite{LabhallaEtAl} as those functions which can be approximated by a polynomial time computable fast converging Cauchy sequence of polynomials with dyadic rational coefficients. It is also shown that integral and maximum of a function are uniformly polytime computable from such a sequence. 
These results were strengthened and refined in various ways 
by Kawamura, M\"uller, R\"osnick, and Ziegler \cite{KMRZ} 
who studied the uniform complexity of maximisation and integration 
for analytic functions and functions in Gevrey's hierarchy 
in dependence on certain parameters which control the growth of the derivatives 
or the proximity of singularities in the complex plane.

While these results already show that maximisation and integration are polytime computable for a large class of practically relevant functions,
there are many practically relevant functions which are not contained in the class of infinitely differentiable functions with well-behaved derivatives:
\begin{itemize}
\item
For applications in control theory it is often necessary to work with functions which are constructed from smooth functions 
by means of pointwise minimisation or maximisation, and thus differentiability is usually lost.
\item
It is not difficult to show that the class of polytime computable functions in Gevrey's hierarchy 
is not uniformly polytime computably closed (with respect to the representation introduced in \cite{KMRZ}) under division by functions which are uniformly bounded by $1$ from below (see Appendix \ref{Appendix A} for a proof).
\end{itemize}

Also, while for any polytime computable $f$ in Gevrey's hierarchy, the function $h(x) = \max_{t \in [-1, x]} f(t)$ 
is again polytime computable, it is in general no longer smooth.
Thus, assuming $\P \neq \NP$, the question arises whether
$h(x)$ is easy to maximise and, more generally, whether
every function which is obtained from a
polytime computable function in Gevrey's hierarchy
by repeatedly applying the parametric maximisation operator  
$f \mapsto \lambda x.\max_{t \in [-1, x]} f(t)$
is polytime computable.

One of our main contributions is to identify a larger class of 
feasibly approximable functions which supports polytime integration and maximisation
and is closed under a larger set of operations, including division and
pairwise and parametric maximisation.

\paragraph{Compositional evaluation strategies.}

In practice, functions of interest are usually constructed from a small set of
(typically analytic) basic functions
by means of certain algebraic operations, such as arithmetic operations, taking primitives, or taking pointwise maxima.
In other words, most functions of practical interest can be expressed symbolically as terms in a certain language. 
Our main observation is that there is such a language which is rich enough to arguably contain the majority of functions of practical
interest, yet restrictive enough to ensure that all functions which are expressible in this language admit uniformly polytime computable
integral, maximum, and evaluation.

To make this claim precise, we introduce the notion of ``compositional evaluation strategy'' for a structure $\Sigma$.
To motivate this notion, consider how a user might specify a computational problem involving real numbers and functions.
We assume that the user specifies the problem symbolically as a term in a certain language and that the end result will be a real number which is expected to be produced to a certain accuracy. 
A library for exact real computation will translate the symbolic representation of the inputs into some internal representation, the details of which will be irrelevant to the user. 
It will operate on the internal representations --- usually in a modular, compositional manner --- 
to eventually produce a name of a real number in the standard representation, which can be queried for approximations to an arbitrary accuracy. 
Thus, there are certain types, such as real numbers in this example, whose representation is relevant to the user, as the user is interested in querying information about them according to a certain protocol, and other types, such as real functions in this example, which are only used internally and whose internal representation can be freely chosen by the library.

\begin{figure}[h]
\centering
\includegraphics[scale = 0.75]{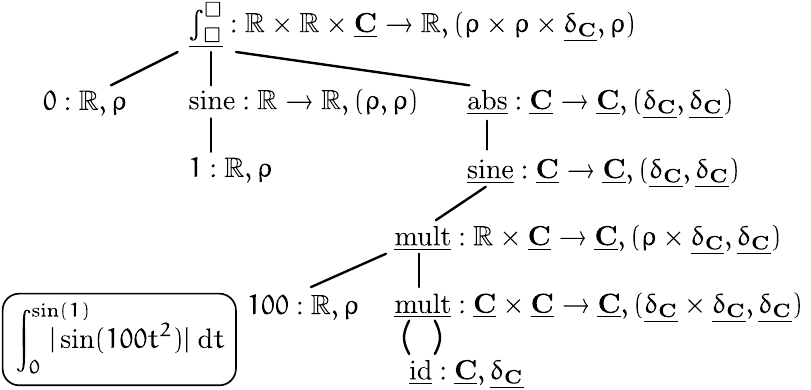}
\caption{Evaluating the term $\int_0^{\sin(1)}\left|\sin(100t^2)\right|\operatorname{dt}$ as a real number.
	     The output is represented in the standard representation $\rho$ of real numbers.
		 The underlined type $\underline{C}$ of real functions is used for internal computations only
		 and its representation $\underline{\delta_{C}}$ can be freely chosen by the library.
	    }
\end{figure}

\noindent
The structures $\Sigma$ we consider consist of:
\begin{enumerate}
\item Fixed spaces: A class of topological spaces with a given representation. 
These kinds of spaces correspond to the kinds of objects which are to be used, among other things, 
as inputs and outputs, so that the kind of information we can obtain on them is fixed. 
\item Free spaces: A class of topological spaces without any given representation. 
These kinds of spaces correspond to the types of intermediate results, whose internal representation is irrelevant to the user.
\item A set of constants and operations on these spaces.
\end{enumerate}

A compositional evaluation strategy provides representations for the free spa\-ces in $\Sigma$ and algorithms, 
in terms of these representations, for all constants and operations in $\Sigma$.
It allows us to evaluate a term in the signature of $\Sigma$ by applying the algorithms in a compositional manner. 
Compositional evaluation can be contrasted with evaluation that involves processing whole terms,
for example, symbolic differentiation.

We say that a compositional evaluation strategy is polytime if it evaluates every term of fixed space type whose free variables are all of fixed space type in polynomial time.
Hence the resource usage of a strategy is measured only in terms of those representations that are relevant to the user. 

Any representation of a space $X$ offers a trade-off between the ability to construct names efficiently and the ability to extract information from names efficiently.
If $\alpha$ and $\beta$ are representations of some space $X$ with $\alpha$ reducing to $\beta$ in polynomial time,
then any function $f \colon X \to Y$ that is polytime when $X$ is represented by $\beta$ is also polytime when $X$ is represented by $\alpha$.
Dually, any function $g \colon Y \to X$ that is polytime when $X$ is represented by $\alpha$ is also polytime when $X$ is represented by $\beta$.
In other words: the higher a representation sits in the reducibility lattice, the fewer functionals and the more points become polytime computable with respect to this representation.
However, the task of evaluating symbolic expressions in a modular manner will usually involve functions of ``symmetric'' type $X \to X$ or $X^n \to X$, such as algebraic operations or closure operations on $X$. 
In general, if $\alpha$ reduces in polynomial time to $\beta$ but not vice versa, then neither does polytime computability of a function $f \colon X \to X$ with respect to $\alpha$ imply polytime computability with respect to $\beta$ nor vice versa.
Thus, polytime reducibility does not allow us to measure how well a given representation trades off the ability to construct names with the ability to extract information from names. 
On the other hand, the study of compositional evaluation strategies will allow us to compare the trade-offs that are offered by different representations.

\paragraph{Results.}

We study various representations of the space $\CI$ 
based on polynomial and rational approximations 
and their relationships in terms of polytime reducibility. 
We show that the representation based on rational approximations 
is polytime equivalent to the representation based on piecewise polynomial approximations
(Corollary~\ref{Corollary: PPoly = Frac}).
This result helps us prove that the class of functions which are representable
by polynomial time computable fast converging Cauchy sequences of piecewise polynomials
is uniformly closed under a set of operations which are typically used
in computing to construct more complicated functions from simpler ones.

In particular, 
we give a compositional evaluation strategy that uses the representation based on approximation by piecewise polynomials 
which evaluates in polynomial time all terms of a structure whose constants are the polytime computable functions in Gevrey's hierarchy
and whose operations include
evaluation, range computation, integration, arithmetic
operations (including division), pointwise and parametric
maximisation, anti-differentiation, composition, and square
roots.

We observe that no compositional evaluation strategy that uses the representations based on polynomial approximation, piecewise affine approximation, or black-box function evaluation can evaluate this structure in polynomial time.
This suggests that when it comes to computing with certain functions of practical interest, the representation based on piecewise polynomial approximations offers a better trade-off between the ability to construct names efficiently and the ability to extract information from names efficiently than other commonly considered representations.

\paragraph{Implementation.}

Whilst in the discrete setting the link between polytime computability
and practical feasibility is
- up to the usual caveats -
well established and confirmed by countless examples of practical implementations,
to our knowledge, little to no work has been done to
link the somewhat more
controversial model of second order complexity in analysis
with practical implementation.
Thus, in order to demonstrate the relevance of our theoretical results to practical computation,
we have implemented compositional evaluation strategies based on the aforementioned representations
for a small fragment of the aforementioned structure within AERN2, 
a Haskell library for exact real number computation. 
We observed that for the most part the benchmark results fit our theoretical predictions quite well.
Our separation results translate to big differences in practical performance,
which can be observed even for moderate accuracies.

This suggests that the latter of the two explanations offered on page \pageref{explanations}
is more applicable:
The infeasibility of maximisation and 
integration with respect to the ``standard representation'' 
of real functions is not a mere accuracy normalisation issue,
and the differences between theoretical predictions and
practical observations are really due to the choice of representation.
The proofs which establish polytime computability translate to algorithms which seem to
be practically feasible, at least up to some common sense optimisations.

\section{The Computational Model}\label{section: Computational Model}

Here we briefly review the basic aspects of the theory of computation with continuous data 
in the tradition of computable analysis, 
as well as the basics of second-order complexity theory.
For background on computability in analysis see \eg \cite{SchroederPhD, PaulyRepresented, Weih, PourElRichards}.
Second-order computational complexity for computable analysis was developed in \cite{KawamuraCook},
building on ideas from \cite{KoFriedman, Ko}.

Let $2 = \{0,1\}$. Let $2^*$ denote the set of all finite binary strings. Let $\B = \left(2^*\right)^{2^*}$ denote \emph{Baire space}
\footnote{In computable analysis it is more common to use the computably isomorphic space ${\N}^\N$
	of functions on the natural numbers,
	but this choice is of course inconsequential.}.
A partial function $f \colon \subseteq \B \to \B$ is called \emph{computable} if 
there exists an oracle Turing machine $M$ which on input 
$u \in \ts$ 
with oracle 
$p \in \dom(f)$ 
computes $f(p)(u) \in \ts$. 
Sometimes, to emphasize the distinction, we will refer to $u$ as the ``input string'' and to $p$ as the ``input oracle'' to $M$.

A \emph{represented space} $(X,\delta_X)$ consists of a set $X$ together with a partial surjection 
$\delta_X \colon \subseteq \B \to X$ 
called the \emph{representation}. 
We will usually write $X$ for $(X,\delta_X)$ if $\delta_X$ is clear from the context. 
A \emph{partial multi-valued function}
$f\colon \subseteq (X,\delta_X) \rightrightarrows (Y, \delta_Y)$ 
between represented spaces 
$(X,\delta_X)$ and $(Y, \delta_Y)$
is just a relation 
$f \subseteq X \times Y$
on the underlying sets.
We write 
$f(x) = \Set{y \in Y}{(x,y) \in f}$
and
$\dom(f) = \Set{x \in X}{f(x) \neq \emptyset}$.
If 
$f\colon \subseteq (X,\delta_X) \rightrightarrows (Y, \delta_Y)$ 
and
$g\colon \subseteq (Y,\delta_Y) \rightrightarrows (Z, \delta_Z)$ 
are partial multi-valued functions, then their \emph{composition}
$g \circ f \colon \subseteq (X, \delta_X) \rightrightarrows (Z, \delta_Z)$
is the partial multi-valued function with
$\dom(g \circ f) = \Set{x \in \dom(f)}{f(x) \subseteq \dom(g)}$
and 
$g\circ f(x) = \bigcup_{y \in f(x)} g(y)$.
If $(X,\delta_X)$ and $(Y, \delta_Y)$ 
are represented spaces, and 
$f\colon \subseteq (X,\delta_X) \rightrightarrows (Y, \delta_Y)$ 
is a partial multi-valued function, we call 
$F \colon \subseteq \B \to \B$ 
a \emph{realiser} of $f$ if 
$\dom(F) \supseteq \dom \left(f \circ \delta_X\right)$ 
and 
$f\left(\delta_X(p)\right) \ni \delta_Y\left(F(p)\right)$ 
for all $p \in \dom\left(f \circ \delta_X\right)$. 
The map $f$ is called \emph{computable} if it has a computable realiser. 
The composition of computable partial multi-valued functions is again computable.
If $X$ carries a topology $\tau$ then $\delta_X \colon \subseteq \B \to X$ is called \emph{admissible for $\tau$} if 
$\delta_X$ is continuous and every continuous map 
$\varphi \colon \subseteq \B \to X$ 
factors through $\delta$ via some continuous
$\Phi \colon \subseteq \B \to \B$, 
\ie $\varphi = \delta_X \circ \Phi$. 
One can show that if $X$ and $Y$ are represented spaces and their respective representations are admissible for topologies on $X$ and $Y$,
then a partial function 
$f\colon \subseteq X \to Y$ 
is sequentially continuous with respect to these representations if and only if it is computable relative to some oracle. 
It was shown by \name{Matthias Schr\"oder} \cite{SchroederPhD, SchroederAdmissibility}
that the class of represented spaces which admit an admissible representation are precisely the 
$\qcb_0$-spaces: 
$T_0$ quotients of countably based spaces. 
The $\qcb_0$ spaces with (sequentially) continuous \emph{total} functions form a Cartesian closed category.
For further details see \cite{SchroederPhD}.

Let us now turn to computational complexity, following the ideas of \name{Kawamura and Cook} \cite{KawamuraCook}.
A string function $\varphi\colon 2^* \to 2^*$ is called \emph{length-monotone} if
\[|u| \leq |v| \rightarrow |\varphi(u)| \leq |\varphi(v)| \]
for all $u,v \in \dom \varphi$. 
If $\varphi$ is a length-monotone function, we define its \emph{size} 
$|\varphi|\colon \N \to \N$ 
via
\[|\varphi|(n) = |\varphi(0^n)|. \]
Note that length-monotonicity implies that 
$|\varphi(u)| = |\varphi(v)|$ whenever $|u| = |v|$, 
which justifies the seemingly arbitrary choice of the string $0^n$ in the definition of the size. 
Let $\M \subseteq \B$ denote the set of length-monotone string functions. 
Note that there is a computable retraction of $\B$ onto $\M$,
so that computability theory remains unaffected by replacing $\B$ with $\M$.
Thus, a mapping 
$f \colon \subseteq \M \to \M$
is computable if there is an oracle Turing machine which on input oracle 
$\varphi \in \dom(f)$, 
and input string $u \in 2^*$ outputs 
$f(\varphi)(u) \in 2^*$. 
The mapping $f$ is 
\emph{computable in time $T\colon \N^{\N} \times \N \to \N$}, 
if there is such a machine which outputs 
$f(\varphi)(u)$ 
within time 
$T(|\varphi|,|u|)$.

We now introduce the class of ``feasibly computable functions'' within this setting.
The set of \emph{second-order polynomials} is defined inductively as follows:
\begin{enumerate}
\item The ``free variable'' $X$ and the ``constant'' $1$ are second-order polynomials.
\item If $P$ and $Q$ are second-order polynomials then so are their sum $P + Q$, their product $P \cdot Q$, 
and the term $\Phi(P)$.
\end{enumerate}

A second-order polynomial $P$ defines a map $\left\llbracket P \right\rrbracket \colon \N^\N \times \N \to \N$ which is inductively defined as follows:
\begin{enumerate}
\item $\left\llbracket 1 \right\rrbracket(f, n) = 1$.
\item $\left\llbracket X \right\rrbracket(f, n) = n$.
\item $\left\llbracket P + Q \right\rrbracket(f, n) = \left\llbracket P \right\rrbracket(f, n) + \left\llbracket Q \right\rrbracket(f, n)$
\item $\left\llbracket P \cdot Q \right\rrbracket(f, n) = \left\llbracket P \right\rrbracket(f, n) \cdot \left\llbracket Q \right\rrbracket(f, n)$
\item $\left\llbracket \Phi(P) \right\rrbracket(f, n) = f(\left\llbracket P \right\rrbracket)$
\end{enumerate}
We will from now on just write $P$ both for the second-order polynomial $P$ and the induced map $\left\llbracket P \right\rrbracket$.

A partial mapping 
$f \colon \subseteq \M \to \M$ 
is called \emph{polytime computable} if 
$f(\varphi)(u)$ 
is computable in time 
$P(|\varphi|,|u|)$ 
for some second-order polynomial $P$.
The class of \emph{total} second-order polytime computable functions coincides with the
class of \emph{basic feasible functionals} \cite{KapronCook}.

These notions translate to represented spaces in the usual way: 
A point $x$ in a represented space $(X,\delta_X)$ is polytime computable if and only if it has a polytime computable name. 
A partial multi-valued function 
$f \colon \subseteq (X,\delta_X) \rightrightarrows (Y, \delta_Y)$ 
is polytime computable if and only if it has a polytime computable $(\delta_X,\delta_Y)$-realiser.
It is often convenient to express the assertion that a function $f \colon X \to Y$ is polytime computable by saying that the value $f(x)$ is uniformly polytime computable in $x$.
The composition of polytime computable functions is again a polytime computable function.
If $X$ is a represented space with representations 
$\delta_X \colon \subseteq \M \to X$ 
and 
$\delta'_X\colon \subseteq \M \to X$ 
we say that $\delta_X$ \emph{reduces to} $\delta_X'$ in polynomial time
and write $\delta_X \leq \delta_X'$ if the identity $\id_X$ on $X$ is polytime
$(\delta_X, \delta_X')$-computable. 
If $\delta_X \leq \delta_X'$ and $\delta_X' \leq \delta_X$ then we say that 
$\delta_X$ and $\delta_X'$ are \emph{polytime equivalent} and write $\delta_X' \equiv \delta_X$.

We will need to introduce canonical representations of finite products. 
Let $\delta_{X_i} \colon \subseteq \M \to X_i$ be a finite family of representations where $i = 1, \dots, n$.
Our goal is to define the product representation
$\delta_{X_1} \times \dots \times \delta_{X_n} \colon \subseteq \M \to X_1 \times \dots \times X_n$ 
Encode the numbers $1, \dots, n$ in binary with a fixed number of digits ($\sim \log_2 n$) and denote the resulting strings by $\mathbf{1}, \dots, \mathbf{n}$. 
Let $\varphi_i \colon \ts \to \ts$ be length-monotone functions for $i = 1, \dots, n$.
Let
\[
	l(k) = \max\Set{|\varphi_j|(k)}{j = 1, \dots, n}.
\]
Define the length-monotone function 
\[
	\langle \varphi_1, \dots, \varphi_n \rangle(\mathbf{i}\cdot u) 
	= \varphi_i(u)\cdot 1 \cdot 0^{l(|u|) - |\varphi_i|(k)} 
\]
Extend this function to all of $2^*$ by letting
$\langle \varphi_1, \dots, \varphi_n \rangle(u) = \varepsilon$,
where $\varepsilon$ denotes the empty string,
if $|u| < |\mathbf{1}|$
and
$\langle \varphi_1, \dots, \varphi_n \rangle(u) = 0^{l\left(|u| - |\mathbf{1}|\right) + 1}$,
if $|u| \geq |\mathbf{1}|$ and $\langle \varphi_1, \dots, \varphi_n \rangle(u)$ was not previously defined.
Now define the representation as follows:
\[\dom\left(\delta_{X_1} \times \dots \times \delta_{X_n}\right) = 
	\Set{\langle \varphi_1, \dots, \varphi_n \rangle}
			{\varphi_i \in \dom\left(\delta_{X_i}\right)} \]
\[\delta_{X_1} \times \dots \times \delta_{X_n}
	\left(\left\langle \varphi_1, \dots, \varphi_n \right\rangle\right) =
		\left(\delta_{X_1}(\varphi_1), \dots, \delta_{X_n}(\varphi_n)\right)
	 \]
	 
Finally, let us give some concrete examples of represented spaces that we will use in the rest of the paper. Countable discrete spaces such as the space of natural numbers $\N$, the space of dyadic rationals $\D$, or the space of rationals $\Q$ are represented via standard numberings, \eg $\nu_\Q \colon \N \to \Q$. By identifying $\N$ with $2^*$, we can view such numberings as maps $\nu_\Q \colon 2^* \to \Q$, which allows us to introduce representations such as $\delta_\Q \colon \M \to \Q$, where $\delta_\Q(\varphi) = \nu_\Q(\varphi(\varepsilon))$.
As a more interesting example, consider the space $\R$ of real numbers. Let
$\rho \colon \subseteq \M \to \R$
with 
$\dom(\rho) = 
	\Set{\varphi \in \M}
	{\forall u,v \in 2^*. 
		\left(\left|\nu_\D\left(\varphi\left(0^{|u|}\right)\right) - \nu_\D\left(\varphi\left(0^{|v|}\right)\right)\right| 
					\leq 2^{-|u|} + 2^{-|v|} \right)}$
and 
$\rho(\varphi) = \lim_{n \to \infty} \nu_\D\left(\varphi\left(0^n\right)\right)$. Using the canonical product construction, we obtain a representation $\rho^n$ of $\R^n$.

\begin{Remark}\label{Remark: compactness and second-order bounds} 
When working with a compact space, one can 
restrict its representation to a compact subset of $\M$,
removing the need for second-order complexity bounds.
Let us illustrate this in the case of the compact unit interval $[-1,1]$.
Using a suitable encoding of dyadic numbers we can find for every real number $x$ a dyadic
approximation of $x$ to error $2^{-n}$ which uses at most 
$2\left(\floor{\log_2(|x| + 1)} + n\right) + 3$
bits. 
Hence, the interval $[-1,1]$ admits a representation $\rho_{[-1,1]} \colon \subseteq \M \to [-1,1]$
with 
$\dom(\rho_{[-1,1]}) \subseteq \Set{\varphi \in \M}{|\varphi|(n) \leq 2\left(n + 1\right) + 3}$.

It is worth noting that we can restrict $\rho$ in a similar way to obtain a representation of all of $\R$, where every name of $x \in \R$ is bounded by $2\left(\floor{\log_2(|x| + 1)} + n\right) + 3$, so that we can bound the running time of an algorithm in terms of the output accuracy and the single number $\log_2(|x| + 1)$ alone, without having to resort to general second-order bounds.

In contrast, the use of genuine second-order bounds cannot be avoided with spaces that are not $\sigma$-compact,
such as $\CI$, the focus of this work.

\end{Remark}

\section{Representations of $\CI$}\label{section: representations}

In this section we introduce a number of commonly used representations of the space $\CI$ of continuous functions over the interval $[-1,1]$ and study their relation in the polytime-reducibility lattice.
Most of these representations and their relationships have been studied already by Labhalla, Lombardi, and Moutai \cite{LabhallaEtAl}, albeit in a slightly different framework. 
Nevertheless, many proofs from \cite{LabhallaEtAl} carry over easily to our chosen framework.
The main new result is the equivalence of rational- and piecewise-polynomial approximations, which is left as an open question in \cite{LabhallaEtAl}.

Most of the representations we study are so-called Cauchy representations,
where an element of a metric space is represented by a fast converging
Cauchy sequence of elements from a countable dense subset.
To spell it out explicitly:

\begin{Definition}
Let $X$ be a separable metric space.
Let $A \subseteq X$ be a countable dense subset of $X$.
Let $\nu_A \colon \ts \to A$ be a numbering of $A$.
Then the Cauchy representation of $X$ induced by $\nu_A$ is the representation
of $X$ where a length-monotone string function $\varphi \in \M$ is a name of 
$x \in X$ if and only if  for all $u \in \ts$ we have $d(\nu_A(\varphi(u)), x) < 2^{-|u|}$.
\end{Definition}

\begin{Definition}\label{Definition: representations}
We define representations $\Poly$, $\PPoly$, $\Frac$, $\PFrac$, $\PAff$, and $\Fun$ of the space $\CI$ of continuous functions over the interval $[-1,1]$ as follows:
\begin{enumerate}
\item 
A $\Fun$-name of a function $f \in \CI$ is a length-monotone string function 
$\varphi\in \M$ 
such that $\varphi(\cdot)$ encodes a sampling of $f$ on dyadic rational points 
and $|\varphi|(\cdot)$ encodes a modulus of uniform continuity of $f$. 
More explicitly, we require
\[ \left|\nu_\D\left(\varphi(\langle u, v \rangle)\right) - f(\nu_\D(u))\right| \leq 2^{-|v|}, \]
where $\langle \cdot, \cdot \rangle$ denotes a standard pairing function on binary strings,
and for all $x, y \in [-1,1]$:
\[|x - y| < 2^{-\left|\varphi\right|(n)} \Rightarrow |f(x) - f(y)| < 2^{-n}. \]
\item 
A $\Poly$-name of a function $f \in \CI$ is a fast converging Cauchy sequence of polynomials 
in the monomial basis with dyadic rational coefficients.
More formally, fix a standard numbering
$\nu_{\D[x]} \colon \ts \to \D[x]$
of the polynomials with dyadic rational coefficients.
The representation $\Poly$ is the Cauchy representation
induced by $\nu_{\D[x]}$.
\item A piecewise polynomial with dyadic rational breakpoints and coefficients is a continuous function $g \colon [-1,1] \to \R$ such that there exist dyadic rational numbers $-1 = a_0, a_1, \dots, a_n = 1$ such that $g|_{[a_i, a_{i + 1}]}$ is a polynomial with dyadic rational coefficients.
A $\PPoly$-name of a function $f \in \CI$ is a fast converging Cauchy sequence of piecewise polynomials in the monomial basis with dyadic rational breakpoints and coefficients. 
More formally, fix a standard numbering of the piecewise polynomials with dyadic breakpoints and coefficients and let $\PPoly$ be the Cauchy representation of $\CI$ induced by this numbering.
\item A $\PAff$-name of a function $f \in \CI$ is a fast converging Cauchy sequence of piecewise affine functions with dyadic breakpoints and coefficients. 
Piecewise affine functions are defined analogously to piecewise polynomials.
More formally, fix a standard numbering of the piecewise affine functions with dyadic breakpoints and coefficients and let $\PAff$ be the Cauchy representation of $\CI$ induced by this numbering.
\item 
A $\Frac$-name of a function $f \in \CI$ is a fast converging Cauchy sequence of rational functions with dyadic coefficients. 
A rational function is a quotient of two polynomials whose denominator has no zeroes in $[-1,1]$.
We choose our notation such that every such rational function is given as a quotient of two polynomials
$P,Q \in \D[x]$ 
which is normalised such that $Q(x) \geq 1$ for all $x \in [-1,1]$.
More formally, fix a standard numbering of the rational functions with dyadic coefficients and let $\Frac$ be the Cauchy representation of $\CI$ induced by this numbering.
\item 
A $\PFrac$-name of a function $f \in \CI$ is a fast converging Cauchy sequence of piecewise rational 
functions with dyadic breakpoints and coefficients. 
Piecewise rational functions are defined analogously to piecewise polynomials and piecewise affine functions.
We again require that the denominator of every rational function be bounded from below by $1$.
More formally, fix a standard numbering of the piecewise rational functions with dyadic breakpoints and coefficients and let $\PFrac$ be the Cauchy representation of $\CI$ induced by this numbering.
\end{enumerate}
\end{Definition}

The representation $\Fun$ is the most efficient representation which renders evaluation computable, 
in the sense that it satisfies the following universal property:

\begin{Proposition}[\cite{KawamuraCook}]\label{Proposition: universal property of Fun}
The following are equivalent for a representation of continuous
functions $\delta\colon\subseteq \M \to \CI$:
\begin{enumerate}
\item Evaluation 
\[\eval\colon \CI \times [-1,1] \to \R, \; (f,x) \mapsto f(x) \]
is polynomial-time $(\delta\times\rho,\rho)$-computable.
\item $\delta \leq \Fun$.
\end{enumerate}
\end{Proposition}
\begin{proof}[Proof sketch]
It is easy to see that evaluation is polytime computable with respect to $\Fun$. Hence, if $\delta \leq \Fun$, then evaluation is polytime computable with respect to $\delta$.
Conversely, assume that $\delta$ renders evaluation polytime computable. 
Given a $\delta$-name of a function $f$ we can clearly evaluate $f$ on dyadic rational points in polynomial time, which yields ``half'' a $\Fun$-name of $f$. 
It remains to show that a modulus of continuity of $f$ can be uniformly computed in polynomial time. 
Since $\delta$ renders evaluation polytime computable there exists a second-order polynomial $P$ which bounds the running time of some algorithm which computes $\eval$. Since $[-1,1]$ is compact, we can assume that the running time of the algorithm on input $\langle\varphi,\xi\rangle$, where $\delta(\varphi) = f$, $\rho(\xi) = x$, is bounded by the function $P(|\varphi|, n)$ 
(since the size of $\xi$ can be bounded independently of $\xi$, cf.~ Remark \ref{Remark: compactness and second-order bounds}). 
Since this function bounds the running time of a $(\delta\times\rho,\rho)$-algorithm which computes $\eval(f,\cdot)\colon \R \to \R$, it follows that $P(|\varphi|, \cdot)$ is a modulus of continuity of $f$. 
As $\varphi$ is length-monotone we have $|\varphi|(n) = |\varphi(0^n)|$, so that this modulus of continuity
is uniformly polytime computable in the name $\varphi$.
\end{proof}

\begin{Corollary}
Let $f \colon [-1,1] \to \R$ be a continuous function. Then $f$ has a polytime computable realiser if and only if it has a polytime computable $\Fun$-name.
\end{Corollary}
On the other hand, the representation $\PPoly$ is interesting since it allows for maximisation and integration in polynomial time.
The following result is folklore, see \eg \cite[Algorithm 10.4]{BasuPollackRoy}:
\begin{Theorem}\label{Theorem: polynomial equation solving}
There exists a polytime algorithm which takes as input 
a non-constant dyadic polynomial $P \in \D[x]$,
a rational number $y \in \Q$,
and an accuracy requirement $n \in \N$
and outputs a list of disjoint intervals
$[a_1, b_1], \dots, [a_m, b_m]$
such that 
\begin{itemize}
\item Every interval contains a solution to the equation $P(x) = y$.
\item Every solution to the equation $P(x) = y$ is contained in some interval.
\item Every interval has diameter $\leq 2^{-n}$.
\end{itemize}
\end{Theorem}

\begin{Corollary}\label{Corollary: PPoly maximisation}
The operators
\[\paramax \colon \CI \to \CI, \; f \mapsto \lambda x.\left(\max\Set{f(t)}{t \leq x} \right), \]
\[\max \colon \CI \times \CI \to \CI, \; (f,g) \mapsto \max(f,g) \]
and 
\begin{align*}
	&\operatorname{join} \colon \subseteq [-1,1] \times \CI \times \CI \to \CI,
	\\
	&(a, f, g) 
	\mapsto 
	\lambda x. 
		\begin{cases} f(x) &\text{if }x \leq a \\
					  g(x) &\text{if }x \geq a
		\end{cases},
\end{align*}
where 
$\dom(\operatorname{join}) = \Set{(a,f,g)}{f(a) = g(a)}$,
are uniformly polytime computable with respect to $\PPoly$.
\end{Corollary}
\begin{proof}[Proof idea]
The proof is very elementary but requires a fair amount of easy but cumbersome quantitative estimates of the size of the objects involved in the construction.
We will therefore only sketch the main ideas behind the proof.

All three claims easily reduce to the claim that the respective operation is computable in polynomial time 
when the input is a dyadic polynomial and the output is a fast converging Cauchy sequence of dyadic piecewise polynomials.

To compute $\paramax$ for a given polynomial $f$ on an interval $[a,b]$, first use Theorem \ref{Theorem: polynomial equation solving} to compute a sufficiently good approximation of the set of critical points of $f$ in $[a,b]$.
Use this to find a list of points $a = x_0 < x_1 < \dots < x_m = b$ meeting the following three conditions: Every $x_i$ is close to either a critical point or a boundary point, we have the inequalities $f(a) \leq f(x_0) < f(x_1) < \dots < f(x_m)$, and $f(x_i)$ satisfies $f(x_i) = \sup_{x \leq x_i} f(x)$.

On the open interval $(x_i, x_{i + 1})$ the equation $f(x) = f(x_i)$ has either no solution, \eg if $x_i$ is a saddle point, or exactly one solution, \eg if $x_i$ is a local minimum.
We can use Theorem \ref{Theorem: polynomial equation solving} to find out in polynomial time which is the case, and in case there is a solution, compute this solution in polynomial time to arbitrary accuracy.
Put $c_i = x_i$ if there is no solution, and if there is a solution, let $c_i$ be a sufficiently good approximation to this solution.
We then have an ascending sequence of points
\[
a = x_0 \leq c_0 < x_1 \leq c_1 < \dots < x_{m - 1} \leq c_{m - 1} < x_m = b.
\]
On the intervals of the form $[x_i, m_i]$ a good approximation of $\paramax(f)$ is given by the constant function $f(x_i)$. 
On the intervals of the form $[c_i, x_{i + 1}]$ a good approximation of $\paramax(f)$ is given by $f$.

The computation of the pointwise maximum of two polynomials reduces to the problem of solving the equation
$P(x) - Q(x) = 0$ to sufficient accuracy.

To avoid case distinctions involving boundary points, it is easiest to compute a piecewise polynomial approximation to $\max(P,Q)$ on all of $\R$.
Given two dyadic polynomials $P$ and $Q$, use Theorem \ref{Theorem: polynomial equation solving} to compute intervals 
$[a_1, b_1], \dots, [a_m, bm]$ 
that enclose the solutions to the equation
$P(x) = Q(x)$ on $\R$ to sufficient accuracy.

Then, by construction, on all intervals of the form $[b_i, a_{i + 1}]$ either $P$ is strictly larger than $Q$ or $Q$ is strictly larger than $P$.
We can decide which of these is the case by comparing 
$P\left(\tfrac{b_i + a_{i + 1}}{2}\right)$ 
and 
$Q\left(\tfrac{b_i + a_{i + 1}}{2}\right)$.
This yields a polynomial approximation to $\max(P,Q)$ on all intervals of the form $[b_i, a_{i + 1}]$.
An analogous argument yields a polynomial approximation on the intervals $(-\infty, a_1)$
and $(b_m, \infty)$.

It remains to compute an approximation on intervals of the form $[a_i, b_i]$.
We have already computed a polynomial approximation $f$ to $\max(P,Q)$ on $[b_{i - 1}, a_i]$ 
and another polynomial approximation $g$ to $\max(P,Q)$ on $[b_i, a_{i + 1}]$.
On $[a_i, b_i]$, let the approximation be the linear interpolation of the values $f(a_i)$ in $a_i$ and
$g(b_i)$ in $b_i$.
If $[a_i, b_i]$ is sufficiently small, then $P$ and $Q$ will be very close on $[a_i, b_i]$,
so that this yields a good approximation.

The polytime computability of $\operatorname{join}$ is established using similar ideas.
\end{proof}
Our goal is to fully understand the relationship between the representations we have just introduced with respect to polytime reducibility.

\begin{Proposition}\label{Proposition: Polynomial Lipschitz constant}
There exists a polytime algorithm which takes as input a piecewise rational function $f$ (given by our standard numbering) and returns as output a Lipschitz constant of $f$.
\end{Proposition}
\begin{proof}
If $R(x) = P(x)/Q(x)$ is a rational function with $Q(x) \geq 1$ for all $x \in [-1,1]$, then by the mean value theorem, a Lipschitz constant of $f$ is given by a bound on $R'(x) = \left(P'(x)Q(x) - P(x)Q'(x)\right)/Q(x)^2$ over $[-1,1]$. Since $Q(x)^2 \geq 1$ it suffices to compute a bound on the absolute value of the polynomial $A(x) = P'(x)Q(x) - P(x)Q'(x)$. If $A(x) = \sum_{i = 0}^n a_ix^i$ then $|A(x)| \leq \sum_{i = 0}^n |a_i|$ for all $x \in [-1,1]$. This is clearly computable in polynomial time. If $f$ is a piecewise rational function with pieces $R_1, \dots, R_m$ then a Lipschitz constant for $f$ is given by the maximum of the Lipschitz constants of the $R_i$'s.
\end{proof}

\begin{Proposition}\label{Proposition: obvious reductions}
We have $\Poly \leq \PPoly \leq \PFrac \leq \Fun$, $\PAff \leq \PPoly$, and $\Frac \leq \PFrac$.
\end{Proposition}
\begin{proof}
The reductions $\Poly \leq \PPoly \leq \PFrac$, $\PAff \leq \PPoly$, and $\Frac \leq \PFrac$ are immediate. It hence suffices to show $\PFrac \leq \Fun$. We will use the universal property of $\Fun$ (Proposition \ref{Proposition: universal property of Fun}) to do so, \ie it suffices to prove that a piecewise rational function can be evaluated in a point in polynomial time.

Suppose we are given a piecewise rational function $f$ with dyadic breakpoints and coefficients, 
a point $x \in [-1,1]$ encoded as a $\rho$-name and an accuracy requirement $n \in \N$. 
By Proposition \ref{Proposition: Polynomial Lipschitz constant} we can compute a Lipschitz constant $L$ of $f$ in polynomial time. 
Query the $\rho$-name of $x$ for a dyadic rational approximation $\tilde{x}$ to error $2^{-n - 1}/L$. 
We can determine an interval $[a,b]$ with $\tilde{x} \in [a,b]$ and $f|_{[a,b]} = P/Q$ with $Q \geq 1$ in polynomial time. 
Now, a dyadic rational approximation $\tilde{y}$ to error $2^{-n - 1}$ of $P(\tilde{x})/Q(\tilde{x})$ is computable in polynomial time. We have
\[|\tilde{y} - f(x)| \leq |\tilde{y} - f(\tilde{x})| + |f(\tilde{x}) - f(x)| \leq 2^{-n-1} + L|\tilde{x} - x| \leq 2^{-n}. \]
\end{proof}

Remarkably, the reduction $\Frac \leq \PFrac$ reverses:

\begin{Theorem}[\cite{LabhallaEtAl}]\label{Theorem: Frac = PFrac}
$\Frac \equiv \PFrac$.
\end{Theorem}

The proof of Theorem \ref{Theorem: Frac = PFrac} given in \cite{LabhallaEtAl} relies mainly on Newman's theorem \cite{Newman} on the rational approximability of the absolute value function.
To establish lower bounds in the reducibility lattice we need to employ \emph{Markov's inequality}. 
For a proof of Markov's inequality see \eg\cite{Cheney}.

\begin{Lemma}[Markov's inequality]\label{Lemma: Markov's inequality}
Let $P$ be a polynomial of degree $\leq n$ on the interval $[-1,1]$. Then
\[\norm{P'} \leq n^2 \norm{P}. \]
On the interval $[a,b]$ we hence have
\[\norm{P'} \leq \frac{2n^2}{b - a} \norm{P}.\]
\end{Lemma}

\begin{Proposition}\label{Proposition: Poly and PAff}
We have $\Poly \not \leq \PAff$ and $\PAff \not \leq \Poly$.
\end{Proposition}
\begin{proof}
The absolute value function $|x|$ is trivially polytime $\PAff$-computable. By Markov's inequality, it is not polytime $\Poly$-computable:
Assume that $(P_n)_n$ is a sequence of polynomials such that $\left|P_n(x) - |x|\right| < 2^{-n}$ for all $n \in \N$. Then on the interval $[-1,0]$ we have $P_n(x) + x < 2^{-n}$ and on the interval $[0,1]$ we have $P_n(x) - x < 2^{-n}$. 
Let $d_n$ denote the degree of $P_n \pm x$.
Applying Markov's inequality to the polynomial $P_n(x) + x$ on the interval $[-1,0]$ yields:
\[|P_n'(x) + 1| \leq 2d_n^2\norm{P_n(x) - |x|} \leq d_n^22^{-n + 1}. \]
Applying the inequality to $P_n(x) - x$ on $[0,1]$ yields:
\[|P_n'(x) - 1| \leq 2d_n^2\norm{P_n(x) - |x|} \leq d_n^22^{-n + 1}. \]
If $d_n \in o(2^n)$ then this implies that $P_n'(0)$ converges to $1$ and $-1$ at the same time, which is absurd. It follows that the size of $(P_n)_n$ grows exponentially in $n$. In particular, $(P_n)_n$ cannot be polytime computable.

For the converse direction we show that the polynomial $x^2$ does not have a polynomial size $\PAff$-name. Consider a piecewise linear approximation $L$ to $x^2$ to error $2^{-n}$ with breakpoints $x_1, \dots, x_m$ and values $y_1, \dots, y_m$. We have 
$|y_i - x_i^2| < 2^{-n}$, and hence for all $t \in [0,1]$:
\[|(1 - t) y_i + t y_{i + 1} - (1 - t) x_i^2 - t x_{i + 1}^2| < 2^{-n}.\] We may hence assume without loss of generality that $y_i = x_i^2$. 
Consider a segment $[x_i, x_{i + 1}]$. 
We have
\begin{align*}
2^{-n} &\geq \left|L - x^2\right| \\
&\geq \left|L\left(\tfrac{1}{2} x_i + \tfrac{1}{2} x_{i + 1}\right) - \left(\tfrac{1}{2} x_i + \tfrac{1}{2} x_{i + 1}\right)^2\right| \\
&= \left|\tfrac{1}{2} x_i^2 + \tfrac{1}{2} x_{i + 1}^2 - \left(\tfrac{1}{2} x_i + \tfrac{1}{2} x_{i + 1}\right)^2\right| \\ 
&= \frac{\left(x_{i + 1} - x_i\right)^2}{4}.
\end{align*}
Now, there exists a segment $[x_i, x_{i + 1}]$ 
with $|x_{i + 1} - x_i| \geq \tfrac{2}{m}$. 
It follows that $m \geq \sqrt{2}^{n}$.
\end{proof}

Together with a result which is proved in the next section (Corollary \ref{Corollary: PPoly = Frac}), we arrive at a complete overview of the reducibility lattice:

\begin{Theorem}\label{Theorem: reducibility between representations}
The following diagram shows all reductions between the representations introduced, up to taking the transitive closure:

\begin{center}
\begin{tikzcd}
\Poly \arrow{dr} \\
	 & \PPoly \arrow{r} & \Frac \arrow{r} \arrow{l} & \PFrac \arrow{l} \arrow{r} & \Fun \\
\PAff \arrow{ur}	 
\end{tikzcd}
\end{center}
No arrow reverses unless indicated.
\end{Theorem}
\begin{proof}
Proposition \ref{Proposition: obvious reductions} establishes the more obvious reductions. 
Proposition \ref{Proposition: Poly and PAff} implies that $\PPoly$ does not reduce to either $\PAff$ or $\Poly$, for any such reduction would establish a reduction from $\Poly$ to $\PAff$ or vice versa.
The reduction $\PPoly \leq \Frac$ follows immediately from $\PFrac \equiv \Frac$. 
The converse is Corollary \ref{Corollary: PPoly = Frac} in Section \ref{Section: division}. 
To see that $\Fun \not \leq \PFrac$, consider the
family of functions $2^{-n}\sin(2^{n}\pi x)$. It is clearly uniformly polytime $\Fun$-computable in $n$, but not uniformly polytime $\Frac$-computable:
Any approximation to the function $2^{-n}\sin(2^{n}\pi x)$ on $[-1,1]$ to error $2^{-n-1}$ has at least $2^n$ zeroes, so that
any rational approximation to this error has a numerator of degree at least $2^{n}$.
\end{proof}

The class of polytime computable points with respect to the representation $\Poly$ has a useful analytic characterisation which was proved by Labhalla, Lombardi, and Moutai \cite{LabhallaEtAl} and strengthened by Kawamura, M\"uller, R\"osnick, and Ziegler \cite{KMRZ}.
For $B > 0$, $\ell > 0$, and $\gamma > 0$ let 
\[\Gev(B, \ell,\gamma) = \Set{f \in \Cinf}{\norm{f^{(n)}} \leq B\cdot \ell^n\cdot n^{\gamma n}} \]
denote the set of Gevrey's \cite{Gevrey} functions of level $\gamma$ with growth parameters $B$ and $\ell$. Note that $\ell = 1$ corresponds to the class of analytic functions. 
The results in \cite{LabhallaEtAl, KMRZ} imply in particular that the above hierarchy collapses on $\Gev(B, \ell,\gamma)$ for all fixed $B$, $\ell$, and $\gamma$:

\begin{Theorem}[\cite{LabhallaEtAl, KMRZ}]\label{Theorem: Poly = Fun on Gev(M,R,a)}
Let $B$, $\ell$, and $\gamma$ be fixed. On $\Gev(B, \ell,\gamma)$ we have
\[\Poly \equiv \PPoly \equiv \Frac \equiv \PFrac \equiv \Fun.\]
\end{Theorem}
\begin{proof}[Proof sketch]
It suffices to show that $\Fun \leq \Poly$. Given a $\Fun$-name of a function $f \in \Gev(B, \ell,\gamma)$, compute a polynomial approximation via Chebyshev interpolation. Since the Chebyshev interpolation is a near-best approximation and $f$ can be approximated efficiently by polynomials, the number of nodes we need in order to compute a polynomial approximation to error $2^{-n}$ is bounded polynomially in $n$. Since we know the constants $B$, $\ell$, and $\gamma$, we can choose the right number of nodes in advance. See \cite[Proposition 21 (e), Theorem 23 (b)]{KMRZ} for details. Also note that the proof in \cite{KMRZ} establishes a much stronger uniform result, where $B$, $\ell$, $\gamma$ are not fixed but given as part of the input.
\end{proof}

\begin{Corollary}\label{Corollary: Gevrey function Fun-polytime iff Poly-polytime}
Let $f \in \Gev(B, \ell,\gamma)$ for some positive constants $B, \ell, \gamma$. 
Then $f$ is polytime computable if and only if it has a polytime computable $\Poly$-name.
\end{Corollary}

\section{Bounded division for piecewise polynomials}\label{Section: division}

We now establish the reduction $\Frac \leq \PPoly$
by giving a polytime division algorithm for piecewise polynomials.
The algorithm will first compute a linear interpolation of the divisor
and then employ an iteration to improve the approximation.
As we cannot evaluate the divisor to infinite precision,
we have to use the following notion:
Let $f \colon [-1,1] \to \R$ be a continuous function. 
Let $x_1, \dots, x_m \in [-1,1]$. 
A \emph{linear $\varepsilon$-interpolation} of 
$f$ at $x_1, \dots x_m$ is a piecewise linear function $L$ with breakpoints 
$x_1, \dots, x_m$ 
which satisfies
$|L(x_i) - f(x_i)| < \varepsilon$.

\begin{Algorithm}[Bounded Division]\label{Algorithm: PPoly division}\hfill
\begin{itemize}
\item 
	Input: A non-constant polynomial $P \in \D[x]$ with $P(x) \geq 1$ on $[-1,1]$. 
	An accuracy requirement $n \in \N$.
\item Output: A piecewise polynomial approximation to $1/P$ on $[-1,1]$ to error $2^{-n}$.
\item Procedure: 
\begin{itemize}
\item Compute a Lipschitz constant $\ell$ of $P$ using Proposition \ref{Proposition: Polynomial Lipschitz constant} and use it to compute an upper bound on the range of $P$ of the form $[1, 2^r]$ for some $r \in \N$.
\item Use Theorem \ref{Theorem: polynomial equation solving} to compute interval upper bounds on the solutions to the equations
\begin{align*}
P'(x) &= 0, \\
P(x)  &= 2^{k} \hspace{5em} \text{for $0 \leq k \leq r$,} \\
P(x)  &= 2^{k + 2}/3 \hspace{3em} \text{for $0 \leq k < r$,}
\end{align*} 
to error $2^{-r - 3}/\ell$.
By this we mean a list of intervals such that each interval contains a solution, each solution is contained in an interval, and each interval has diameter at most $2^{-r - 3}/\ell$.
\item Sort the intervals together with the boundary points (viewed as degenerate intervals) in ascending order to get a list 
\[[-1,-1] = I_1 < I_2 < \dots < I_m = [1,1].\] 
If two intervals should overlap, refine them such that they are either disjoint or their union has diameter smaller than $2^{-r - 3}/\ell$. In the latter case replace them with their union.
\item Compute a linear $2^{-r - 4}$-interpolation $Q_0$ of $1/P$ at the centres of the intervals.
\item Let $N = \lceil\log_2(3n)\rceil$.
\item For $k = 1, \dots, N$:
\begin{itemize}
\item Put $Q_{k + 1} = 2Q_k - PQ_k^2$.
\end{itemize}
\item Output $Q_N$.
\end{itemize}
\end{itemize}
\end{Algorithm}

\begin{Remark}\hfill
\begin{enumerate}
\item The iteration employed in Algorithm \ref{Algorithm: PPoly division} is the well-known Newton-Raphson division method.
\item While, by Lemma \ref{Lemma: PPoly division polytime} below, Algorithm \ref{Algorithm: PPoly division} already runs in polynomial time, its practical performance can be improved significantly by employing, within the iteration,
size-reduction techniques
such as degree reduction and sweeping, maintaining rigorous error bounds.
\item The resource usage of Algorithm \ref{Algorithm: PPoly division} is mainly dominated by the multiplication of polynomials with potentially large degree within the Newton-Raphson iteration. 
While the degrees can sometimes be kept small by the aforementioned size-reduction techniques, 
there are practical instances of the problem where the degrees grow quite large, 
resulting in poor practical performance, despite the algorithm being polytime. 
For more details, see Section \ref{Section: Experiments}.
\item If $P \in \D[x]$ is any non-constant polynomial with $P(x) \geq b > 0$ on $[-1,1]$, we can apply Algorithm \ref{Algorithm: PPoly division} to $P/b$ and use it to compute an approximation to 
$1/P(x) = (1/b)\left/(P(x)/b)\right.$. If we know that $P(x) > 0$, without knowing a bound, we can use Corollary \ref{Corollary: PPoly maximisation} to find a lower bound $b > 0$, but since we need to witness that $b$ is above $0$, the complexity depends additionally on $\log_2(\inf_{x \in [-1,1]} P(x))$. 
\end{enumerate}
\end{Remark}

\begin{figure}[h]
\centering
\includegraphics[scale = 0.63]{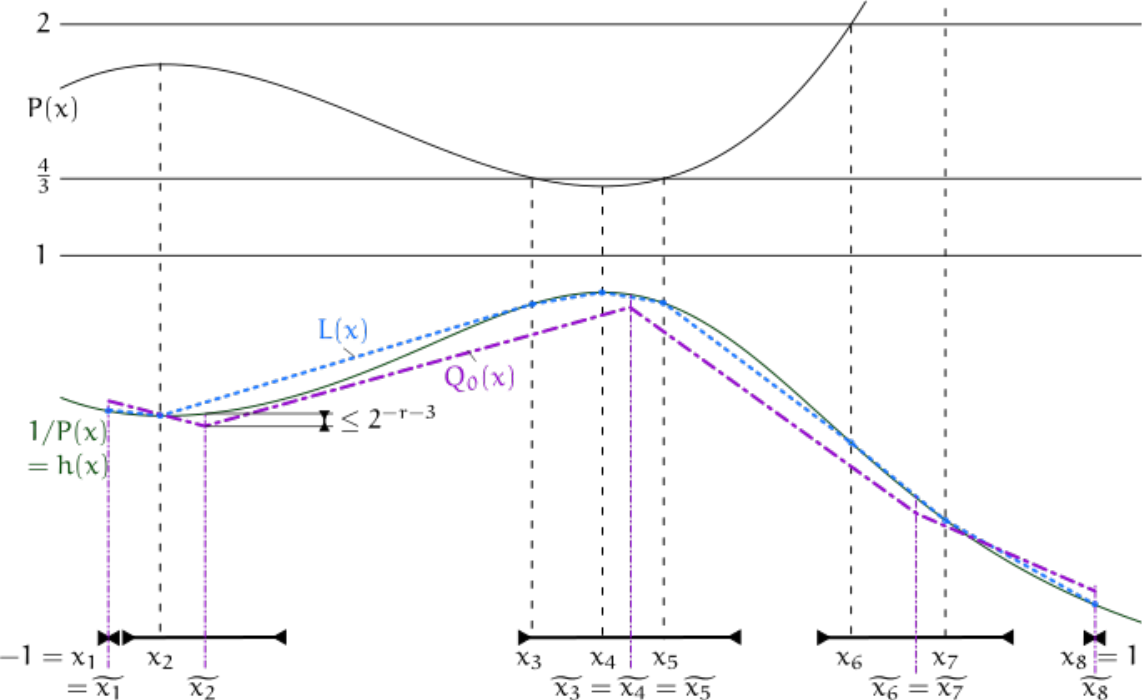}
\caption{Overview of the notation used in the correctness proof of Algorithm \ref{Algorithm: PPoly division}
(Lemma \ref{Lemma: correctness of bounded division}).}
\end{figure}

\begin{Lemma}\label{Lemma: correctness of bounded division}
Algorithm \ref{Algorithm: PPoly division} is correct.
\end{Lemma}
\begin{proof}
Let $-1 = a_1 < a_2 < \dots < a_m = 1$ be the union of the boundary points and the zeroes of $P'(x)$, sorted in an increasing order, so that $1/P$ is monotone on each interval $[a_i, a_{i + 1}]$. 
On $[a_i, a_{i + 1}]$, let
\[a_i = b^i_1 < b^i_2 < \dots < b^i_{k_i} = a_{i + 1} \]
be the solutions of the equations $P(x) = 2^k$ and $P(x) = 2^{k + 2}/3$, where $k \in [0, r]$, together with the boundary points. 
Let
\[-1 = x_1 < x_2 < \dots < x_l = 1 \]
denote the $b^i_j$'s, sorted in an increasing order. 
Let $L$ be the linear interpolation of $1/P$ in the $x_i$'s. 

The proof relies on the following two inequalities:
\begin{itemize}
\item \textbf{Claim 1:} $|L(x) - 1/P(x)| < 1/(2P(x))$ for all $x \in [-1,1]$.
\item \textbf{Claim 2:} $|Q_0(x) - L(x)| \leq 1/(4P(x))$ for all $x \in [-1,1]$.
\end{itemize}
We prove by induction on $k$ the inequality $|Q_k(x) - 1/P(x)| \leq (3/4)^{2^{k}} \cdot (1/P(x))$
for all $x \in [-1,1]$.
The base case is established by combining the above claims using the triangle inequality.
The induction step is given below:
\begin{align*}
|Q_{k + 1}(x) - 1/P(x)| &= |2Q_k(x) - P(x)Q_k(x)^2 - 1/P(x)|
\\  &= |P(x)| \cdot |2Q_k(x)/P(x) - Q_k(x)^2 - (1/P(x))^2| 
\\  &= |P(x)|\cdot|Q_k(x) - 1/P(x)|^2 
\\  &\leq (3/4)^{2^{k + 1}} \cdot (1/P(x)).
\end{align*}
Using the definition $N = \lceil\log_2(3n)\rceil$ 
we obtain $|Q_N(x) - 1/P(x)| \leq 2^{-n}$
which finishes the proof.

\paragraph{Proof of Claim 1.}
We claim that $|L(x) - 1/P(x)| < 1/(2P(x))$ for all $x \in [-1,1]$. Consider an interval of the form $[x_i, x_{i + 1}]$. Since $1/P$ is monotone on the interval, we have
\[|L(x) - 1/P(x)| \leq |1/P(x_{i}) - 1/P(x_{i + 1})|. \]
If $x_i$ and $x_{i + 1}$ are inner points of the interval $[-1,1]$ then there are four cases:
\begin{enumerate}
\item $P(x_i) = 2^k$, $P(x_{i + 1}) = 2^{k + 2}/3$. We have:
\[|1/P(x_{i}) - 1/P(x_{i + 1})| = |2^{-k} - 3\cdot2^{-k - 2}| = 2^{-k - 2}. \]
Since $P$ is monotonically increasing, we have:
\[1/(2P(x)) \geq 1/(2P(x_{i + 1})) = \tfrac{3}{2}2^{-k - 2} \geq 2^{-k - 2}. \]
\item $P(x_i) = 2^k$, $P(x_{i + 1}) = 2^{(k - 1) + 2}/3$. We have:
\[|1/P(x_i) - 1/P(x_{i + 1})| = |2^{-k} - 3\cdot2^{-k - 1}| = 2^{-k - 1}. \]
Since $P$ is monotonically decreasing, we have:
\[1/(2P(x)) \geq 1/(2P(x_i)) = 2^{-k - 1}. \]
\item $P(x_i) = 2^{k + 2}/3$, $P(x_{i + 1}) = 2^{k + 1}$. We have:
\[|1/P(x_i) - 1/P(x_{i + 1})| = |3\cdot 2^{-k - 2} - 2^{-k - 1}| = 2^{-k - 2}. \]
Since $P$ is monotonically increasing, we have:
\[1/(2P(x)) \geq 1/(2P(x_{i + 1})) = 2^{-k - 2}. \]
\item $P(x_i) = 2^{k + 2}/3$, $P(x_{i + 1}) = 2^{k}$. We have:
\[|1/P(x_i) - 1/P(x_{i + 1})| = |3\cdot2^{-k-2} - 2^{-k}| = 2^{-k - 2}. \]
Since $P$ is monotonically decreasing, we have:
\[1/(2P(x)) \geq 1/(2P(x_i)) = \tfrac{3}{2}2^{-k-2} \geq 2^{-k - 2}. \]
\end{enumerate}
The cases where $x_i$ or $x_{i + 1}$ is a boundary point are treated similarly.

\paragraph{Proof of Claim 2.}
We claim that $|Q_0(x) - L(x)| < 1/(4P(x))$ for all $x \in [-1,1]$.
By construction
every $x_i$ is contained in some interval $I_j$ which is computed by Algorithm \ref{Algorithm: PPoly division}.
Conversely every interval $I_j$ contains some $x_i$.
Let $\widetilde{x}_i$ denote the centre of the interval $I_j$ which contains $x_i$.
Note that different $x_i$'s could yield equal $\widetilde{x}_i$'s.

As both $L$ and $Q_0$ are piecewise linear, the distance $|L(x) - Q_0(x)|$ attains its maximum in one of the $x_i$'s or one of the $\widetilde{x}_i$'s.

Let us introduce some notation to improve the readability of the following estimates.
Write $h(x) = 1/P(x)$. 
Write $\varepsilon_x = 2^{-r - 4}/\ell$ for the distance between $x_i$ and $\widetilde{x}_i$.
Write $\varepsilon_y = 2^{-r - 4}$ for the distance between $Q_0(\widetilde{x}_i)$ and $h(\widetilde{x}_i)$.

We find:
\begin{align*}
	|Q_0(\widetilde{x}_i) - L(\widetilde{x}_i)|
	&\leq |Q_0(\widetilde{x}_i) - h(\widetilde{x}_i)|
	+ |h(\widetilde{x}_i) - L(x_i)| + |L(x_i) - L(\widetilde{x}_i)| \\
	&= |Q_0(\widetilde{x}_i) - h(\widetilde{x}_i)| + |h(\widetilde{x}_i) - h(x_i)| + |L(x_i) - L(\widetilde{x}_i)| \\
	&\leq \varepsilon_y + \ell \varepsilon_x + \ell \varepsilon_x\\
	&\leq 2^{-r - 4} + 2^{-r - 3} \\
	&\leq \tfrac{1}{4P(x)}
\end{align*}
The last line uses that $r$ is by definition an upper bound on $\log_2 P(x)$.
The estimate of the second factor in the third-to-last line uses the fact that any Lipschitz constant for $h$
is also a Lipschitz constant for $L$.
Note that since $P$ is bounded by $1$ from below, any Lipschitz constant for $P$ on $[-1,1]$ is also a Lipschitz constant for $1/P$ on $[-1,1]$.

To estimate $|Q_0(x_i) - L(x_i)|$ we need to find a bound on the Lipschitz constant of $Q_0$.
As $Q_0$ is piecewise linear, it suffices to compute a number $\ell_Q$ satisfying
\[
	|Q_0(\widetilde{x}_i) - Q_0(\widetilde{x}_{i + 1})| \leq \ell_Q |\widetilde{x}_i - \widetilde{x}_{i + 1}|
\]
for all $i$.

If $\widetilde{x}_i = \widetilde{x}_{i + 1}$ then any non-negative $\ell_Q$ will do.
Hence let us assume that $\widetilde{x}_{i} \neq \widetilde{x}_{i + 1}$.
Then by construction $|\widetilde{x}_i - \widetilde{x}_{i + 1}| > 2 \varepsilon_x$.
We calculate:
\begin{align*}
	|Q_0(\widetilde{x}_i) - Q_0(\widetilde{x}_{i + 1})| 
	&\leq
	 |Q_0(\widetilde{x}_i) - h(\widetilde{x}_{i})| 
	 + |h(\widetilde{x}_i) - h(\widetilde{x}_{i + 1})| 
	 + |h(\widetilde{x}_{i + 1}) - Q_0(\widetilde{x}_{i + 1})| \\
	&\leq 
	2 \varepsilon_y + \ell |\widetilde{x}_i - \widetilde{x}_{i + 1}| \\
	&\leq 
	\left(\tfrac{\varepsilon_y}{\varepsilon_x} + \ell\right)  |\widetilde{x}_i - \widetilde{x}_{i + 1}|.
\end{align*}

We now obtain:
\begin{align*}
	|Q_0(x_i) - L(x_i)|
	&\leq 
	|Q_0(x_i) - Q_0(\widetilde{x}_i)|
	+	|Q_0(\widetilde{x}_i) - h(\widetilde{x}_i)|
	+	|h(\widetilde{x}_i) - h(x_i)|
	+   |h(x_i) - L(x_i)|\\
	&\leq 
	\left(\tfrac{\varepsilon_y}{\varepsilon_x} + \ell\right) \varepsilon_x
	+	\varepsilon_y
	+   \ell \varepsilon_x\\
	&= 2 \varepsilon_y + 2\ell \varepsilon_x \\
	&= 2^{-r-3} + 2^{-r-3} \\
	&\leq \tfrac{1}{4P(x)}
\end{align*}
\end{proof}

Let us now show that Algorithm \ref{Algorithm: PPoly division}
runs in polynomial time.
The following lemma ensures that the initial approximation
can be computed in polynomial time:

\begin{Lemma}
There exists a polytime algorithm which takes as input a $\Fun$-name of 
a function $f \in \CI$,
a list of points $x_1, \dots, x_m \in [-1,1]$, 
and an error bound $\Q \ni \varepsilon > 0$, 
and returns as output a linear $\varepsilon$-interpolation of $f$ at $x_1, \dots, x_m$.
\end{Lemma}

\begin{Lemma}\label{Lemma: PPoly division polytime}
Algorithm \ref{Algorithm: PPoly division} runs in polynomial time.
\end{Lemma}
\begin{proof}
The size of the Lipschitz constant $\ell$ of $P$ is bounded polynomially in the degree and the size of its coefficients. The bound $[1, 2^r]$ on the range can be given as $r = \lceil\log_2 (\ell + 1)\rceil$. 
Hence there are only polynomially many equations to solve, and since the algorithm in Theorem \ref{Theorem: polynomial equation solving} runs in polynomial time, the overall complexity of the construction of the initial approximation $Q_0$ is polynomial. 
In particular, the number of segments of $Q_0$ is polynomial in the size of $P$. 
The degree of the $k^{\text{th}}$ approximation is $(2^{k} - 1)\deg P + 2^k$, 
so the degree of the $N^\text{th}$ approximation is in
$O \left((6n + 1)\deg P + 6n\right)$, 
which is polynomial in the size of $P$ and $n$. 
The number of segments does not change during the iteration.

It remains to estimate the size of the coefficients.
For a polynomial $A$, encoded as a list of dyadic rational numbers in standard notation,
let $t_A$ denote the number of terms of $A$, \ie $t_A = \deg A + 1$,
and let $c_A$ (by abuse of notation)
denote the bitsize of the coefficients of the given encoding of $A$.
Let $c_k = c_{Q_k}$ and $t_k = t_{Q_k}$. 
We have $t_k = \deg(Q_k) + 1 = (2^{k} - 1)\deg P + 2^k + 1$. 
If $A$ and $B$ are polynomials, then $c_{AB} \leq c_A + c_B + \min\{t_A,t_B\}$ so that 
\[c_{Q_k^2} \leq 2c_k + t_k \]
and hence
\[c_{k + 1} = c_{2Q_k - PQ_k^2} \leq \max\left\{c_k + 1, c_P + (2c_k + t_k) + \min\{t_k,t_P\}\right\} \leq c_P + 2c_k + 2t_P. \]
it follows by induction that
\[c_k \leq (2^k - 1)c_P + 2^kc_0 + 2(2^k - 1)t_P\]
Hence we have:
\[c_N \in O \left( (n + 1)c_P + n c_0 + 2(n + 1) t_P \right) \]
which is polynomial in $c_P$ and $n$.
\end{proof}

By applying Algorithm \ref{Algorithm: PPoly division} piece-by-piece we obtain:

\begin{Theorem}\label{Theorem: PPoly division}
Bounded division,
\[\div \colon \subseteq \CI \times \CI \to \CI, \; (f,g) \mapsto f / g, \]
where
\[\dom(\div) = \Set{(f,g) \in \CI \times \CI}{g(x) \geq 1 \text{ for all } x \in [-1,1]} \]
is uniformly $(\PPoly,\PPoly)$-polytime computable. 
\end{Theorem}

\begin{Corollary}\label{Corollary: PPoly = Frac}
$\PPoly \equiv \Frac$.
\end{Corollary}
\begin{proof}
Suppose we are given a fast converging sequence $(P_n(x)/Q_n(x))_n$ of rational
functions which converge to $f \colon [-1,1] \to \R$,
normalised such that $Q_n(x) \geq 1$ on $[-1,1]$.
Apply Algorithm \ref{Algorithm: PPoly division} to obtain a piecewise polynomial
approximation $g_n$ to $P_{n + 1}(x)/Q_{n + 1}(x)$ to error $2^{-n - 1}$.
Then the sequence $(g_n)_n$ is a fast converging sequence of piecewise polynomials
with limit $f$, in other words, a $\PPoly$-name of $f$.
\end{proof}

We also obtain a corollary on the complexity of integrating rationally approximable functions, which is not immediately obvious:

\begin{Corollary}
The integration functional
\[\int\colon C([-1,1]) \times \R \to \R, \; (f, x) \mapsto \int_{-1}^x f(t) dt \]
is uniformly $(\Frac \times \rho,\rho)$-polytime computable.
\end{Corollary}

\section{Compositional Evaluation Strategies}\label{section: strategies}

In this section we introduce the notion of compositional evaluation strategy over an algebraic structure $\Sigma$. 
This will allow us to state our main result on the existence of a modular polytime algorithm for evaluating all sufficiently simple symbolic expressions which involve maximisation or integration.

For a class of spaces $C$, let $\Prod(C)$ denote the class of all finite and countable products of members of $C$, \ie a space $A$ belongs to $\Prod(C)$ if and only if it is of the form $A_1 \times \dots \times A_n$ or $\prod_{i \in \N} A_i$ with $A_i$ being members of $C$.

Consider structures of the form 
\[\Sigma = \left(\Fix, \Free, \Op, \Const \right)\] 
where

\begin{enumerate}
\item 
	$\Fix$ is a set of represented spaces $(Y, \delta_Y)$,
	containing at least the space $\left(\N, \delta_{\N}\right)$
	of natural numbers with the standard representation induced
	by the binary notation.
\item $\Free$ is a set of represented spaces.
\item $\operatorname{Op}$ is a set of partial multi-valued operations of the form
$f \colon \subseteq A \rightrightarrows B$
where $A, B \in \Prod(\Fix \cup \Free)$.
\item $\operatorname{Const}$ is a subset of the disjoint union of all spaces in $\Prod(\Fix \cup \Free)$.
\end{enumerate}

The set $\Fix$ is called the set of \emph{fixed spaces}, the set $\Free$ is called the set of \emph{free spaces}, the set $\Op$ is called the set of \emph{operations} and the set $\Const$ is called the set of \emph{constants}. An operation of the type 
$A_1 \times \dots \times A_n \rightrightarrows B_1 \times \dots \times B_m$
will be called an $(n,m)$-ary operation. An $(n,1)$-ary operation will also be called an $n$-ary operation for short.

A constant $c \in X$ where $X \in \Prod(\Fix \cup \Free)$ will be called a constant of type $X$ and we write $c \colon X$. For every $X \in \Prod(\Fix \cup \Free)$ we introduce a countable set of free variables $x_n \colon X$ of type $X$. A term over the signature of $\Sigma$ is defined inductively as follows:
\begin{enumerate}
\item Every free variable of type $X$ is a term of type $X$.
\item Every constant of type $X$ is a term of type $X$.
\item If $t_1 : X_1$ and $t_2 : X_2$ are terms, then $(t_1, t_2)$ is a term of type $X_1 \times X_2$.
\item If $t \colon X$ is a term of type $X$ with a free variable $n$ of type $\N$ then $\lambda n. X$ is a term of type $X^{\N}$.
\item If $t \colon X$ is a term and $f \colon \subseteq X \rightrightarrows Y$ is an operation, then $f(t)$ is a term of type $Y$.
\end{enumerate}
A term is called \emph{closed} if it contains no free variables. 
We denote the set of closed terms of $\Sigma$ by $\CT(\Sigma)$. 
If $t \colon X$ is a closed term we denote by 
$\sem{t}_\Sigma$ 
the set of elements of $X$ which it represents under the obvious semantics%
\footnote{
	The application of a \emph{partial} operation could lead to the semantics of a term to be undefined. 
	It is however straightforward to define (inductively) what it means for a term to be well-defined, 
	and we will henceforth assume that all terms are well-defined.}. 
A term $t : Y$ is called 
\emph{semi-closed} 
if it contains no free variables of free space type.
We denote the set of semi-closed terms of $\Sigma$ by $\SCT(\Sigma)$. 
If $x_1 : X_1, \dots, x_n : X_n$ are the free variables in $t$, 
then on the semantic side $t$ defines a partial operation
\[ \sem{t}_\Sigma \colon \subseteq X_1 \times \dots \times X_n \rightrightarrows Y.\]

Suppose we are given a structure $\Sigma$. 
A \emph{compositional evaluation strategy} for $\Sigma$ consists of:

\begin{enumerate}
\item For every free space $X$ of $\Sigma$ a representation $\delta_{X} \colon \subseteq \M \to X$.
\item 
    For each operation $f \colon \subseteq X \rightrightarrows Y$
    of $\Sigma$
	an algorithm which computes a 
	$(\delta_X, \delta_Y)$-realiser of $f$.
\item 
    For each constant $x \colon X$ of $\Sigma$
    an algorithm which computes a $\delta_X$-name of $x$.
\end{enumerate}
A compositional evaluation strategy $S$ defines a map
\[\eval_S \colon \subseteq \CT(\Sigma) \to \M \]
which sends a closed term $t : X$ of type $X$ to a point $\eval_S(t) \in \M$ with $\delta_X(\eval_S(t)) \in \sem{t}_{\Sigma}$. 
We define the \emph{running time of $S$ on $t$}
\[
	T_S(t, \cdot) \colon \N \to \N
\]
as the time it takes to compute $\eval_S(t)(\cdot)$ using the compositional evaluation strategy. 
The map $\eval_S$ extends to a map
\[
	\eval_S \colon \subseteq \SCT(\Sigma) \to \M^\M 
\]
which sends a semi-closed term $t : Y$ to a realiser of the operation $\sem{t}_{\Sigma}$.
The \emph{running time of $S$ on $t \in \SCT(\Sigma)$} - if it exists - is then the smallest second-order function
\[
	T_S(t, \cdot, \cdot) \colon \N^\N \times \N \to \N, 
\]
such that $T_S(t, |\varphi|, |u|)$ is a bound on the time it takes to compute $\eval_S(t)(\varphi, u)$ using $S$.
We say that a strategy $S$ is \emph{polytime} if it evaluates every semi-closed term of $\Sigma$ of \emph{fixed space type} in polynomial time.

It should be noted that a strategy being polytime does not imply that the running time of the strategy 
grows polynomially in the size of the term it is evaluating.
For example, consider the structure $\Sigma = (\{\R\}, \emptyset, \{\texttt{square}\}, \{\Q\})$,
where $\texttt{square}(x) = x^2$ is the squaring operation.
This structure
can be evaluated in polynomial time.
However, when evaluating the term 
\[ 
	\texttt{square}^{(n)}(2) = \underbrace{\texttt{square} \circ \dots \circ \texttt{square}}_{n \text{ times}} (2)
\]
to an accuracy of $1$ bit, the running time of any compositional evaluation strategy for this structure 
grows super-exponentially in $n$.

\section{On the complexity of integration and maximisation for common functions}

Consider the structure 
\[
	\Sigma = (\{\R\}, \{\CI\}, \operatorname{Op}, \operatorname{Const})
\]
where $\operatorname{Const}$ is the disjoint union of all 
polytime computable real numbers 
and all 
polytime computable functions 
in Gevrey's hierarchy 
and $\operatorname{Op}$ consists of the following operations:
\begin{enumerate}
\item $\operatorname{const} \colon \R \to \CI, \; x \mapsto \lambda t. x$.
\item $+\colon \CI \times \CI \to \CI, \; (f,g) \mapsto f + g.$
\item $\times\colon \CI \times \CI \to \CI, \; (f,g) \mapsto f \cdot g.$ 
\item $- \colon \CI \to \CI, \; f \mapsto -f$.
\item $\div \colon \subseteq \CI \times \CI \to \CI, \; (f, g) \mapsto f/g$, where 
\[\dom(\div) = \Set{(f,g) \in \CI \times \CI}{g(x) \geq 1 \text{ for all }x \in [-1,1]}.\]
\item $\sqrt{|\cdot|} \colon \CI \to \CI, \; f \mapsto \sqrt{|f|}.$
\item $\circ \colon \subseteq \CI \times \CI \to \CI, \; (f,g) \mapsto f \circ g,$ where 
\[\dom(\circ) = \Set{(f,g) \in \CI \times \CI}{g([-1,1]) \subseteq [-1,1]}. \]
\item $\max \colon \CI \times \CI \to \CI, \; (f,g) \mapsto \max(f,g).$
\item $\paramax \colon \CI \to \CI, \; f \mapsto \lambda t.\max\Set{f(s)}{s \in [-1,t]}.$
\item 
\begin{align*}
	&\operatorname{join} \colon \subseteq [-1,1] \times \CI \times \CI \to \CI,
	\\
	&(a, f, g) 
	\mapsto 
	\lambda x. 
		\begin{cases} f(x) &\text{if }x \leq a, \\
					  g(x) &\text{if }x \geq a,
		\end{cases}
\end{align*}
where 
$\dom(\operatorname{join}) = \Set{(a,f,g)}{f(a) = g(a)}$.
\item $\primit \colon \CI \to \CI, \; f \mapsto \lambda t.\int_{-1}^t f(s) \operatorname{ds}.$
\item $\eval \colon \CI \times [-1,1] \to \R, \; (f, x) \mapsto f(x).$
\end{enumerate}

Note in particular that $\Sigma$ allows us to express the integral 
$\int_{a}^b f(x) \operatorname{dx}$
as 
\[
\eval(\primit(f), b) - \eval(\primit(f), a)
\]
and the maximum 
$\max_{x \in [a,b]} f(x)$
as 
\[
 \eval(\paramax(\operatorname{join}(a, \const(\eval(f,a)), f)), b).
\]
The structure $\Sigma$ arguably contains most univariate functions on a compact interval that are used in practical computing,
as it contains the polytime analytic functions and all commonly available closure operations.

\begin{Theorem}\label{Theorem: evaluating all common functions in polynomial time}
	There exists a compositional evaluation strategy for $\Sigma$,
	using $\PPoly$ to represent the space $\CI$, 
	that runs in polynomial time.
\end{Theorem}
\begin{proof}
Let $f$ be a polytime computable function in Gevrey's hierarchy.
Then $f$ has a polytime computable $\Fun$-name by Proposition \ref{Proposition: universal property of Fun}.
It follows from Theorem \ref{Theorem: Poly = Fun on Gev(M,R,a)}
that $f$ has a polytime computable $\PPoly$-name.

It remains to show that the operations listed above are polytime computable with respect to $\PPoly$.
Polytime computability of the first four operations is obvious.
Polytime computability of $\div$ is proved in Theorem \ref{Theorem: PPoly division}.
Polytime computability of composition is easily established for $\Frac$, which is polytime equivalent to $\PPoly$ by Corollary \ref{Corollary: PPoly = Frac}.
The polytime computability of $\sqrt{|\cdot|}$, 
follows from Newman's Theorem \cite{Newman} on the rational approximability of the square root (see \cite{LabhallaEtAl} for details) 
in conjunction with the polytime computability of division and the polytime computability of composition.
The polytime computability of $\max$, $\paramax$, and $\operatorname{join}$ is established in 
Corollary \ref{Corollary: PPoly maximisation}.
The polytime computability of $\primit$ is elementary.
The polytime computability of $\eval$ is established in 
Proposition \ref{Proposition: obvious reductions}.
\end{proof}

Theorem \ref{Theorem: evaluating all common functions in polynomial time} can be taken as evidence that there are no ``natural'' functions whose integral and maximum are difficult to compute.

\begin{Theorem}\label{Theorem: no polytime evaluation strategy based on Fun or Poly}
	There is no evaluation strategy which uses 
	the representations $\Poly$, $\PAff$, or $\Fun$
	which evaluates $\Sigma$ in polynomial time.
\end{Theorem}
\begin{proof}
	Consider the problem of computing 
	$\int_{-1}^1 |x| \operatorname{dx}$
	which can be expressed by the term 
	$
		\eval(\primit(\max(-x,x)), 1)
	$
	of $\Sigma$.
	
	Any correct algorithm that sends a $\Fun$ name of a function $f$ to a Cauchy name of the real number $\int_{-1}^1 f(x) \operatorname{dx}$ has to query its input function at least $2^{\omega_f(n)}$ times,
	where $\omega_f$ is the modulus of continuity provided by the $\Fun$ name of $f$,
	to produce an approximation to error $2^{-n}$.
	A fortiori any compositional evaluation strategy using $\Fun$ 
	requires running time at least $2^{n}$ when evaluating the term
	$\eval(\primit(\max(-x,x)), 1)$ to error $2^{-n}$.
	This shows that no compositional evaluation strategy using $\Fun$ evaluates $\Sigma$ in polynomial time.

	Any correct algorithm that sends a $\Poly$ name of a function $f$ to a Cauchy name of the real number $\int_{-1}^1 f(x) \operatorname{dx}$ has to query its input for a polynomial approximation to $f$ to error at least $2^{-n}$ in order to compute an approximation to the output to error $2^{-n}$.
	But it was shown in the proof of Proposition \ref{Proposition: Poly and PAff} that the size of any sequence of polynomial approximations to $|x|$ grows exponentially in the accuracy of the approximation.
	This shows that no compositional evaluation strategy using $\Poly$ evaluates $\Sigma$ in polynomial time. 

	To show the analogous claim for the representation $\PAff$ consider the term 
	$
		\eval(\primit(x^2), 1)
	$
	which represents the number $\int_{-1}^1 x^2 \operatorname{dx}$
	and use that, by the proof of Proposition \ref{Proposition: Poly and PAff}, any $\PAff$ name of $x^2$ grows exponentially. 
\end{proof}

Compare Theorems \ref{Theorem: evaluating all common functions in polynomial time} and \ref{Theorem: no polytime evaluation strategy based on Fun or Poly} with Theorem \ref{Theorem: reducibility between representations}.
By Theorem \ref{Theorem: reducibility between representations} there is a strict linear chain of polytime reductions 
\[
	\Poly < \PPoly < \Fun.
\]
Intuitively this says that among the three representations $\Poly$ contains the greatest amount of information about a function while $\Fun$ contains the least, with $\PPoly$ being somewhere in the middle.
By Theorem \ref{Theorem: no polytime evaluation strategy based on Fun or Poly} and its proof, the representation $\Poly$ contains too much information to evaluate all terms of the structure $\Sigma$ in polynomial time, as it does not render sufficiently many points of $\CI$ polytime computable.
By contrast, the representation $\Fun$ contains too little information to evaluate all terms of $\Sigma$ in polynomial time, as it does not render sufficiently many functionals on $\CI$ polytime computable.

By Theorem \ref{Theorem: evaluating all common functions in polynomial time} and its proof, the representation $\PPoly$ does evaluate all terms of $\Sigma$ in polynomial time, which can be intuitively interpreted as saying that $\PPoly$ contains just the right amount of information to evaluate $\Sigma$ efficiently.   

\section{Experiments}
\label{Section: Experiments}

We describe a set of experiments we conducted to gauge the practical
efficiency of the representations $\Fun$, $\Poly$, $\PPoly$, $\Frac$
as well as some more efficient variants:
\begin{itemize}
  \item $\BFun$ represents a function $f{:{}}[-1,1]\to\R$ by $F{:{}}\ID[-1,1] \to \ID$
  , where $\ID$ is the discrete space of intervals with dyadic rational endpoints,
  such that \(f(x) = \bigcap\{F(X)\, |\, x\in X\in \ID[-1,1]\}\) for each $x\in[-1,1]$.
  \item $\DBFun$ represents a continuously differentiable function $f$ by a pair
  $F,F'$ where $F$ is a $\BFun$ name of $f$ and $F'$ is a $\BFun$ name of $f'$.
  \item ``Local'' representation $\LPoly$ that represents $f$ by
  a dependent-type function $F$ that maps each $D\in\ID$ to a Poly-name of $f|_D$.
  Representations $\LPPoly$ and $\LFrac$ are defined analogously. 
\end{itemize} 

The representation $\BFun$ is the standard representation of continuous functions
in interval analysis.
Our benchmarks confirm that it is much more efficient than $\Fun$ from a practical perspective.
The main reason why we use $\Fun$ instead of $\BFun$ in our theoretical considerations
is that $\BFun$ is not a well-behaved representation from the point of view of 
second-order complexity, as the size function of a name does not provide
sufficient information on the ``complexity'' of that name.
In fact, it is easy to show that every computable function has a 
polytime computable $\BFun$-name.
On the other hand, the use of $\Fun$ is justified by Proposition \ref{Proposition: universal property of Fun}.
We consider DBFun, although it is not a representation of continuous functions,
because it alleviates one of the disadvantages that Fun and BFun have compared to
polynomial-based representations, namely the in-ability to utilise the potential smoothness of $f$.
The ``local'' representations are polytime equivalent to their ``global'' counterparts, 
so that we did not have to consider them in the theoretical part of this paper.
However, it is obvious that they offer a great practical advantage,
as it would be wasteful to compute an approximation 
over the whole interval $[-1,1]$ when only a local approximation
over a small interval is needed.

For each representation, we implemented a calculator for the following task:

\begin{description}
\item{\textbf{Input}:} A real function $x\mapsto f(x)$ given as 
a symbolic expression over a signature with the functions \(x\mapsto 1\), \( x \mapsto x\)
and pointwise sine, cosine, maximum, and field operations

\item{\textbf{Output}:} $\max_{x\in[-1,1]}f(x)$ or $\int_{-1}^1 f(x)\, \mathrm{d}x$ 
encoded as a fast converging Cauchy sequence
\end{description}

Note that the input and output are independent of the chosen function representation.
Thus all the calculators have the same ``user interface''.
  
The input expressions are evaluated bottom-up using 
an evaluation strategy based on the chosen representation.  
E.g., on input $\sin(\sin(x))$ the $\Poly$-calculator 
constructs a polynomial approximation of $\sin(x)$ and feeds this approximation
again to the same implementation of sine that produces a polynomial 
approximation of $\sin(\sin(x))$. The calculators do not attempt to simplify, differentiate 
or otherwise symbolically manipulate the given expression.

In other words, we implement compositional evaluation strategies for the structure
\[\Sigma = \left(\left\{(\R,\rho)\right\}, \left\{\CI\right\}, \left\{\range, \int, +, \times, -, \div, \sin, \cos, \max \right\}, \{1,x\} \right). \]
based on the different representations.
Theorems \ref{Theorem: evaluating all common functions in polynomial time}
and \ref{Theorem: no polytime evaluation strategy based on Fun or Poly}  suggest that representations based on $\PPoly$ 
will perform best in our benchmarks.
In particular, they should perform better than representations based on $\Fun$ for almost any function.
They should also perform better than representations based on $\Poly$ for non-smooth functions.
Our experimental findings confirm this for the majority of functions we have considered.

\subsection{Implementation}

Due to space constraints we describe only the most significant aspects of our
implementation.
We describe it in more detail in the technical paper \cite{fnreps-technical}.
The source code\footnote{\url{http://tinyurl.com/aern2-fnreps}} is available online.
It should be emphasized that our implementation is not designed to outperform practical algorithms for integration and range computation,
but to provide a common framework to offer a fair comparison of different algorithmic approaches.
Our implementation framework is not optimised for speed, and, with the exception of $\DBFun$, we do not exploit any information about the derivatives of our functions,
which in practice makes an enormous difference.

\paragraph{Fun representations.}

Most operations over $\Fun$, $\BFun$ and $\DBFun$ are implemented in a straightforward manner
ball/point-wise.  Range maximum and integration are implemented using bisection.
The target accuracy of integration is raised by 1 bit with each domain bisection.
Integration bisection ends when the area of the ``box'' enclosing the function over the 
segment is below the target accuracy.
The maximisation algorithm employs a simple branch and bound method
to prune away intervals where the maximum is not attained.
The derivative available in DBFun is used solely to improve the interval extension
of $f$ using the formula $f([c\pm \varepsilon])\subseteq f(c)\pm \varepsilon\cdot f'([c\pm \varepsilon])$. 

\paragraph{Polynomial representations.}

Polynomials are represented primarily sparsely in the Chebyshev basis
over $[-1,1]$ with dyadic coefficients.  Any terms that are smaller than the
current accuracy target are sweeped away, \ie removed and their size added to the error radius.
The choice of the Chebyshev basis is motivated by the fact that this sweeping procedure works 
well in the Chebyshev basis, but not in the monomial basis.
While our theoretical results are formulated with respect to the monomial basis,
it is straightforward to verify that the translations between the Chebyshev basis and the monomial basis are computable in polynomial
time.
The range maximisation algorithm combines the root counting techniques described in
Chapter 10 of \cite{BasuPollackRoy}
with a branch and bound method similar to the one employed in the 
maximisation algorithm for $\BFun$.
It temporarily translates the polynomials to a dense representation 
in the monomial basis with integer coefficients.

Poly division, pointwise maximisation, and for very large polynomials also multiplication, 
is computed using an interval version of Chebyshev interpolation for analytic functions
via the encoding of discrete cosine transform (DCT) from \cite{BT97}.

$\PPoly$ division is described in Section \ref{Section: division}.
$\PPoly$, $\Frac$, and local representations use essentially the same algorithm as $\Poly$
for range maximisation.
Frac integration is computed via a translation to PPoly.

The local representations delegate integration to their non-local
counterparts over the equidistant partition of the domain into
$n$ segments where $n$ is the target accuracy
\footnote{\ie the required error bound is $2^{-n}$.}. 

\subsection{Benchmarks and results}

\paragraph{Well-behaved analytic functions.}

First, consider the functions in Fig.~\ref{fig:exprs-easy} 
that are analytic on the whole complex plane.
As the charts are linear-logarithmic, exponential maps show
as straight lines and a polynomial maps show as logarithmic
curves.

We have not included timings for representations $\PPoly$, $\Frac$, $\LPPoly$ and $\LFrac$
in Fig.~\ref{fig:exprs-easy}
because for these expressions our implementations of $\PPoly$ and $\Frac$
compute identical approximations as our implementation of $\Poly$.

\begin{figure}[tp]
\begin{center}
\begin{tabular}{cc}
\toprule
\multicolumn{2}{c}{Execution time (s) vs Accuracy (bits)}
\\[1ex]
Range maximum over $[-1,1]$ & Integral over $[-1,1]$
\\
\midrule
\multicolumn{2}{l}{%
$f(x)=\sin(10x)+\cos(20x)$
\includegraphics[width=.3\hsize, height=4ex]{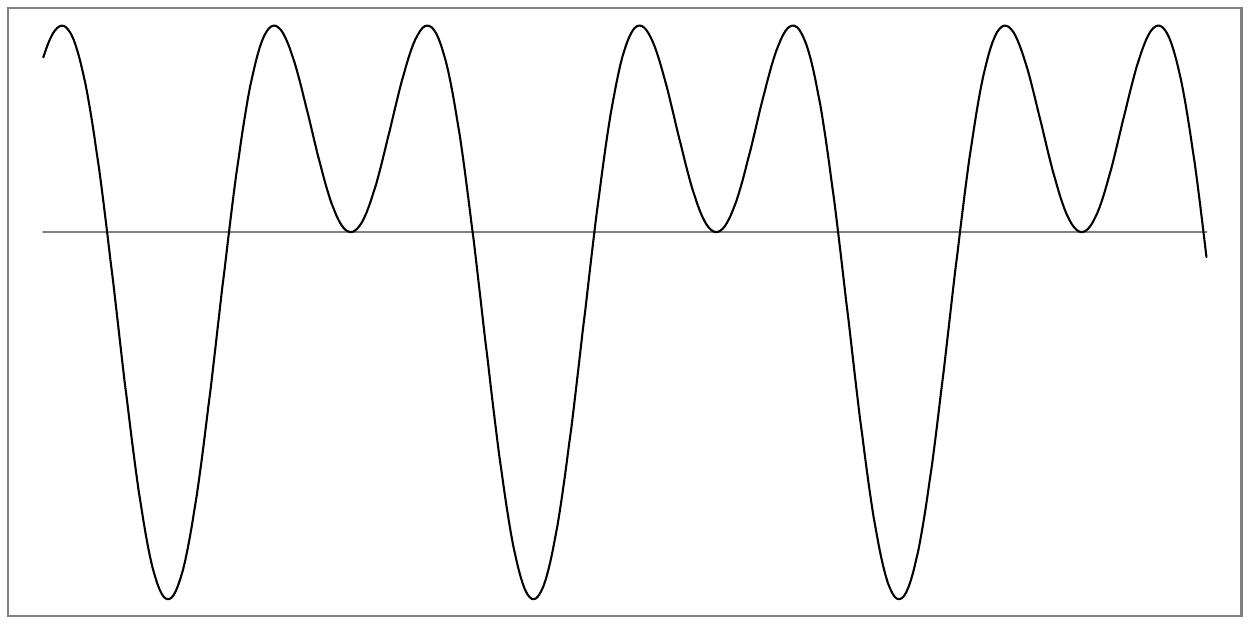}}
\\
\includegraphics[width=.46\hsize]{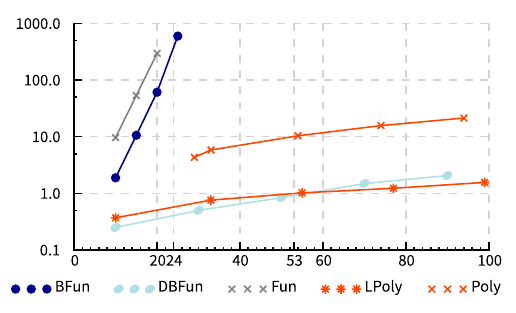}
&
\includegraphics[width=.46\hsize]{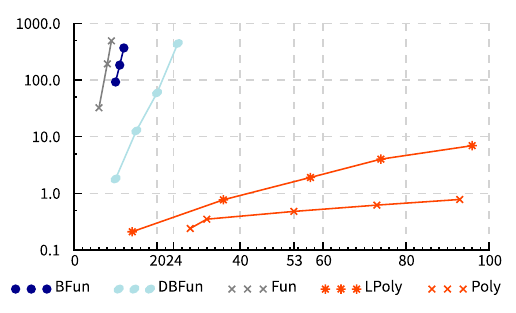}
\\
\midrule
\multicolumn{2}{l}{%
$f(x)=\sin(10x)+\cos(7\pi x)$
\includegraphics[width=.3\hsize, height=4ex]{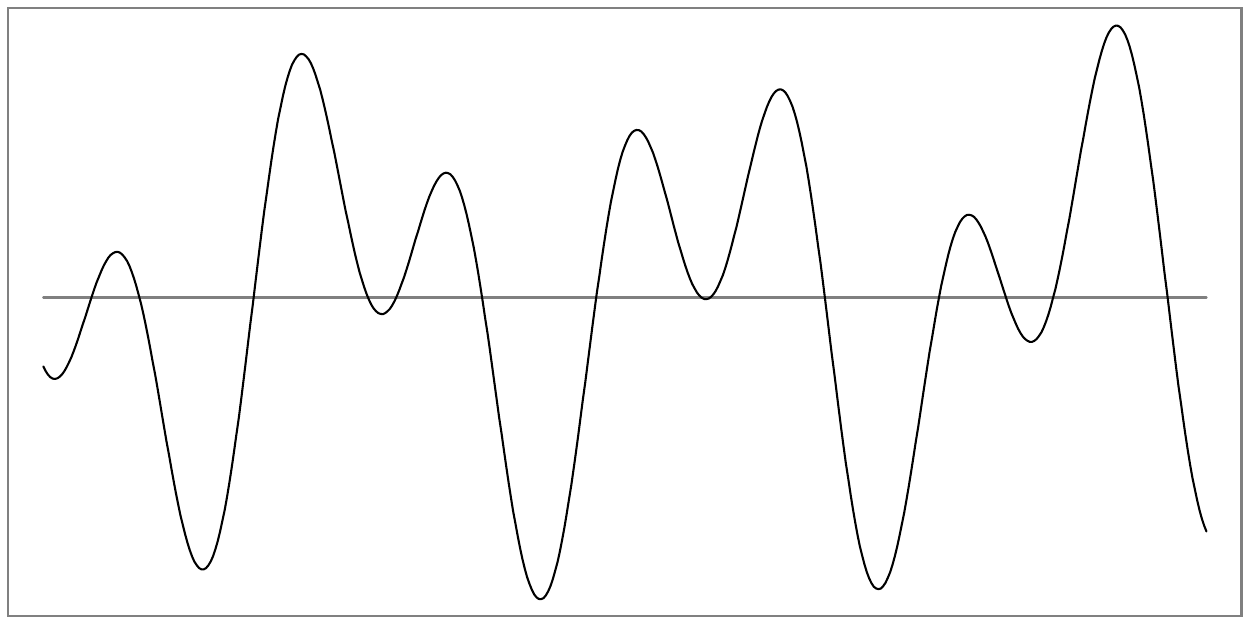}}
\\
\includegraphics[width=.46\hsize]{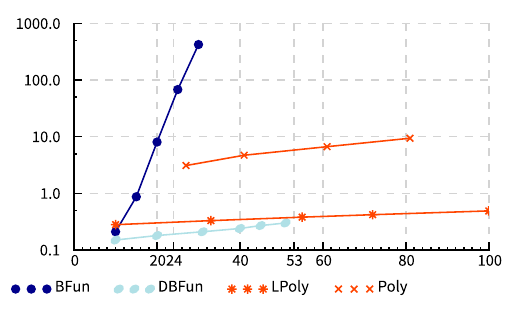}
&
\includegraphics[width=.46\hsize]{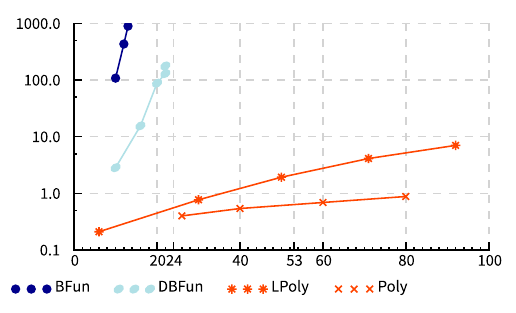}
\\
\midrule
\multicolumn{2}{l}{%
$f(x)=\sin(10x+\sin(7\pi x^2))+\cos(10x)$
\includegraphics[width=.3\hsize, height=4ex]{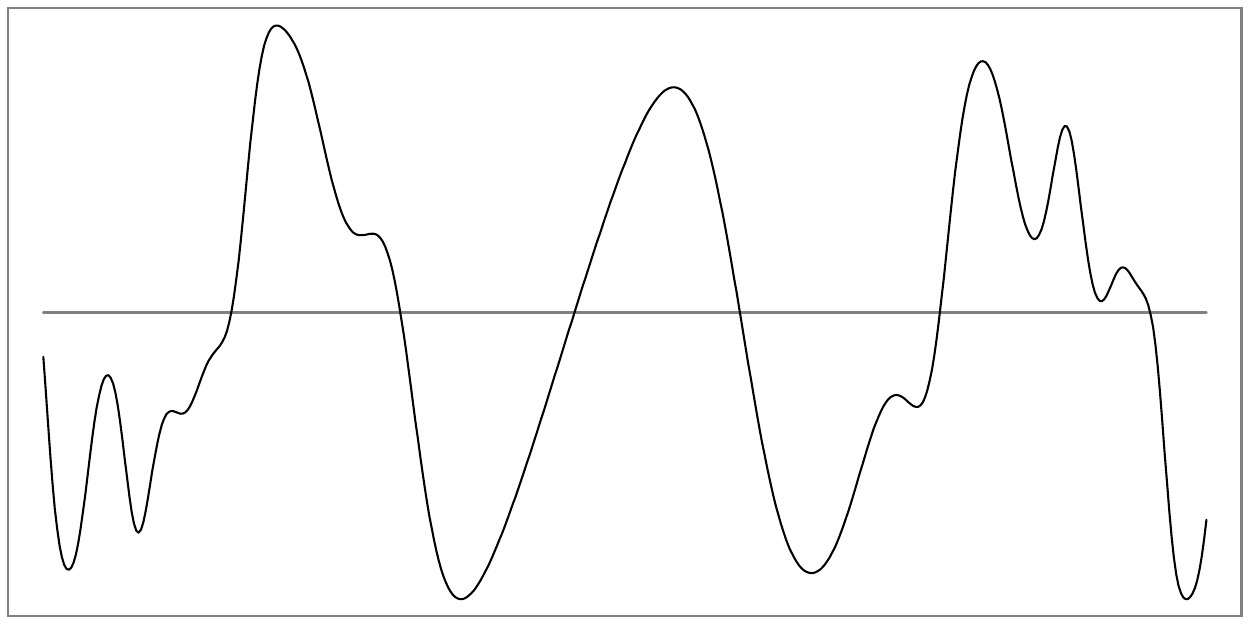}}
\\
\includegraphics[width=.46\hsize]{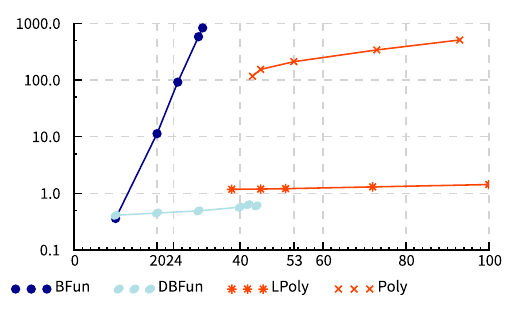}
&
\includegraphics[width=.46\hsize]{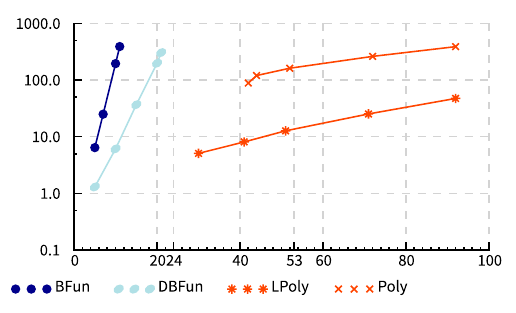}
\\
\bottomrule
\end{tabular}

  \caption{Measurements for analytic functions without nearby singularities}
  \label{fig:exprs-easy}
\end{center}
\end{figure}

$\Fun$ performed so poorly that we struggled to get any points within
the constraints of our charts.  Therefore we applied it on the first and simplest
function only.

$\DBFun$ has computed the range of \(\sin(10x)+\cos(20x)\) much more
efficiently than the range of \(\sin(10x)+\cos(7\pi x)\).
This indicates that $\DBFun$ maximisation is very sensitive
to the quality of the interval extension of $f$.
We expect that $\BFun$ is also sometimes similarly sensitive although we have not
observed it in our benchmarks.

These examples confirm our prediction that range and integral for these kinds of functions
are much more efficient to compute via polynomial approximations than
simply via $\Fun$ representations.  
Moreover, localisation seems to help when functions are defined by
a nested application of elementary functions.

\paragraph{Functions with division and pointwise maximum.}

\begin{figure}[tp]
\begin{center}
\begin{tabular}{cc}
\toprule
\multicolumn{2}{c}{Execution time (s) vs Accuracy (bits)}
\\[1ex]
Range maximum over $[-1,1]$ & Integral over $[-1,1]$
\\
\midrule
\multicolumn{2}{l}{%
$\displaystyle f(x)=\frac{\sin(10x)+\cos(7\pi x)}{100x^2+1}$
\includegraphics[width=.3\hsize, height=5ex]{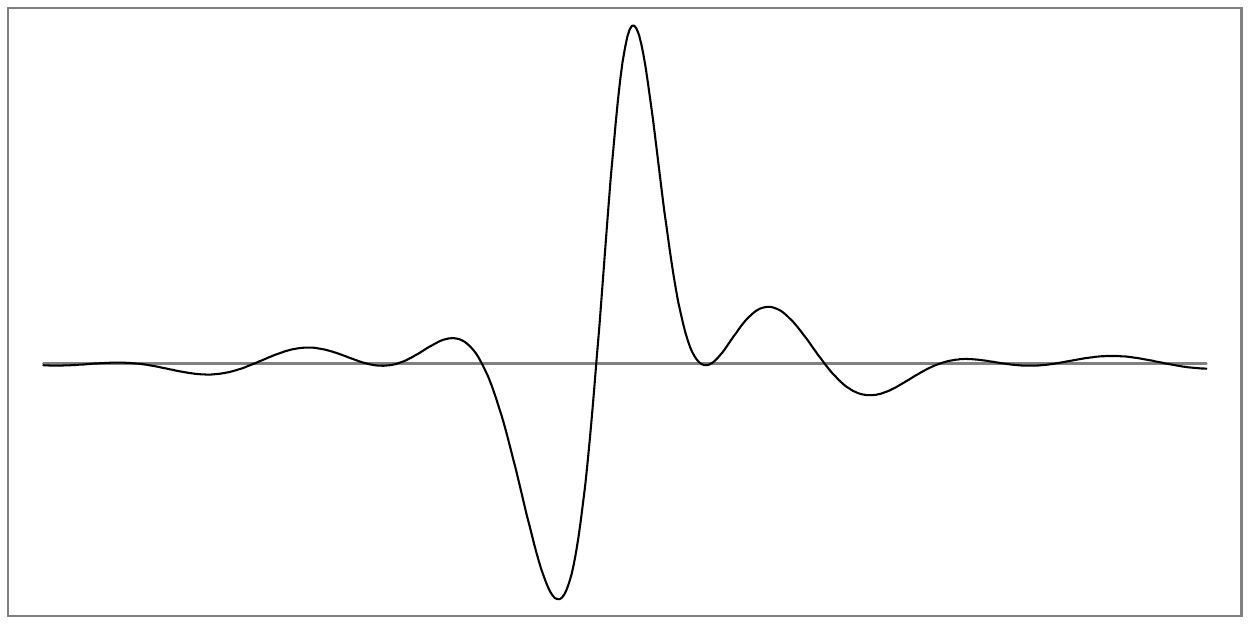}}
\\
\includegraphics[width=.46\hsize]{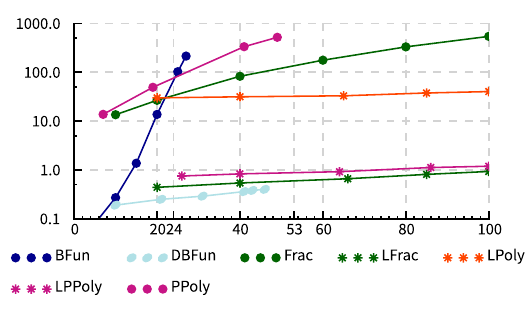}
&
\includegraphics[width=.46\hsize]{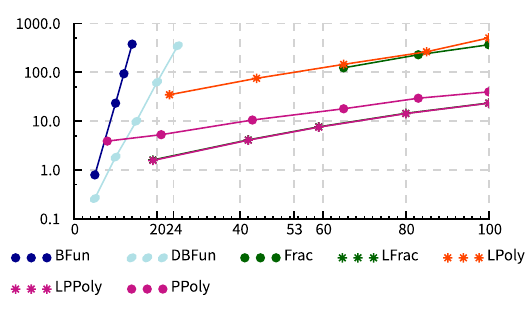}
\\
\midrule
\multicolumn{2}{l}{%
$\displaystyle f(x)=\frac{\sin(10x)+\cos(7\pi x)}{10(\sin(7x))^2+1}$
\includegraphics[width=.3\hsize, height=5ex]{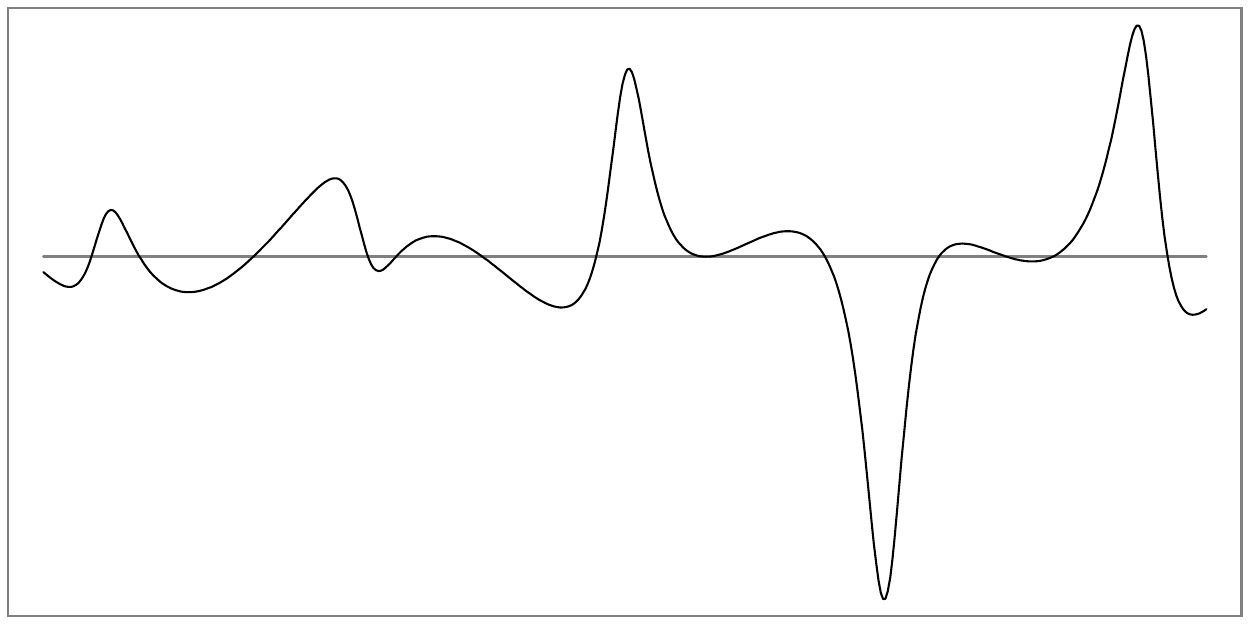}}
\\
\includegraphics[width=.46\hsize]{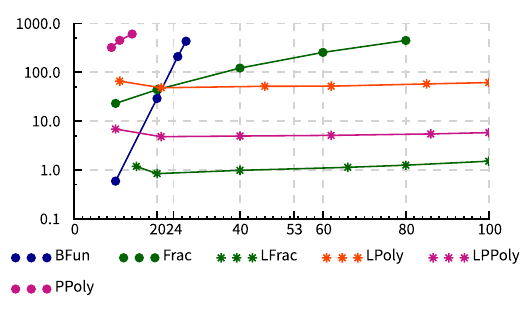}
&
\includegraphics[width=.46\hsize]{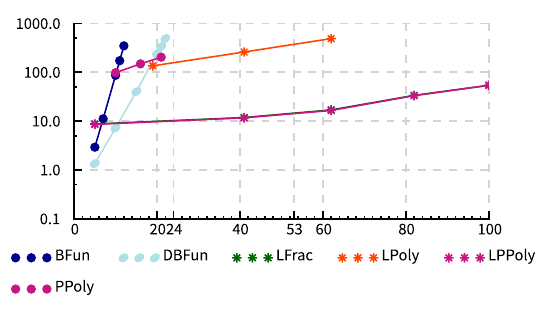}
\\
\midrule
%
%
\multicolumn{2}{l}{%
$\displaystyle f(x)=\max\left(\sin(10x), \cos(11x) \right)$
\includegraphics[width=.3\hsize, height=5ex]{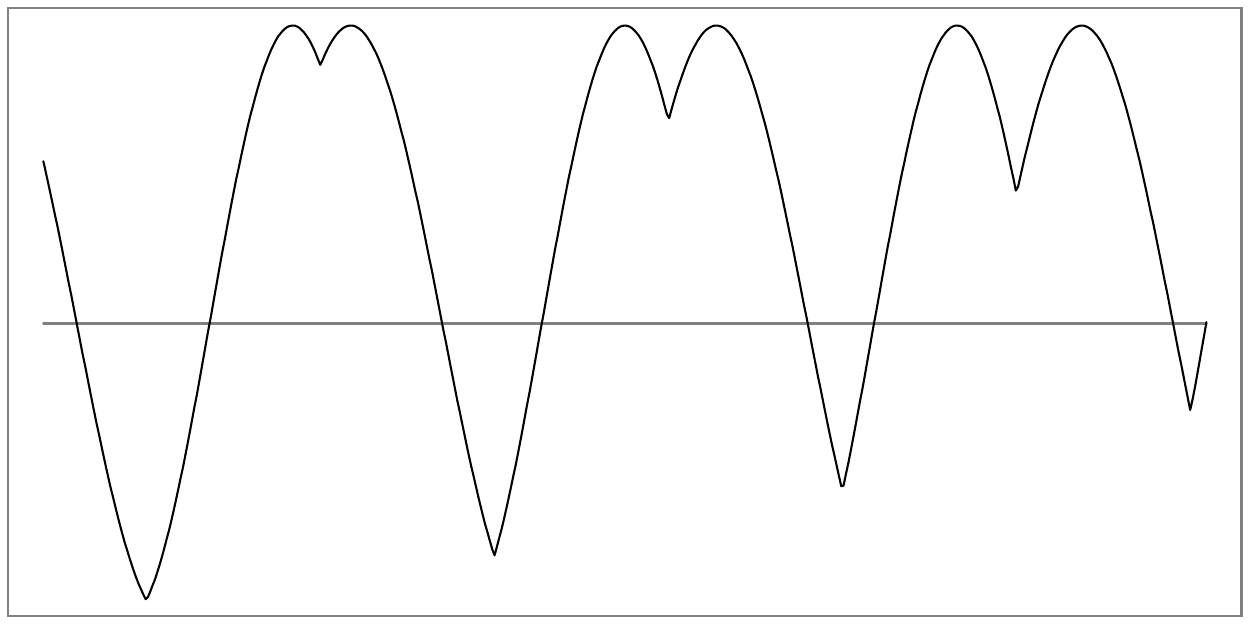}}
\\
\includegraphics[width=.46\hsize]{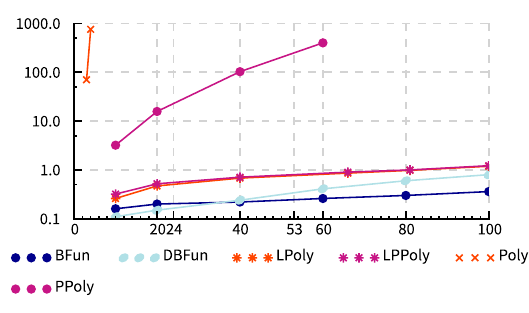}
&
\includegraphics[width=.46\hsize]{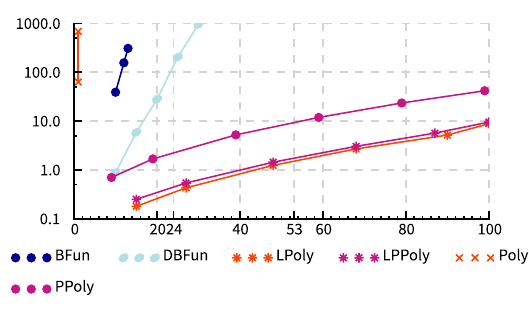}
\\
\midrule
\multicolumn{2}{l}{%
$\displaystyle f(x)=\max\left(\frac{x^2}{2},\frac{\sin(10x)+\cos(7\pi x)}{10(\sin(7x))^2+1}\right)$
\includegraphics[width=.3\hsize, height=5ex]{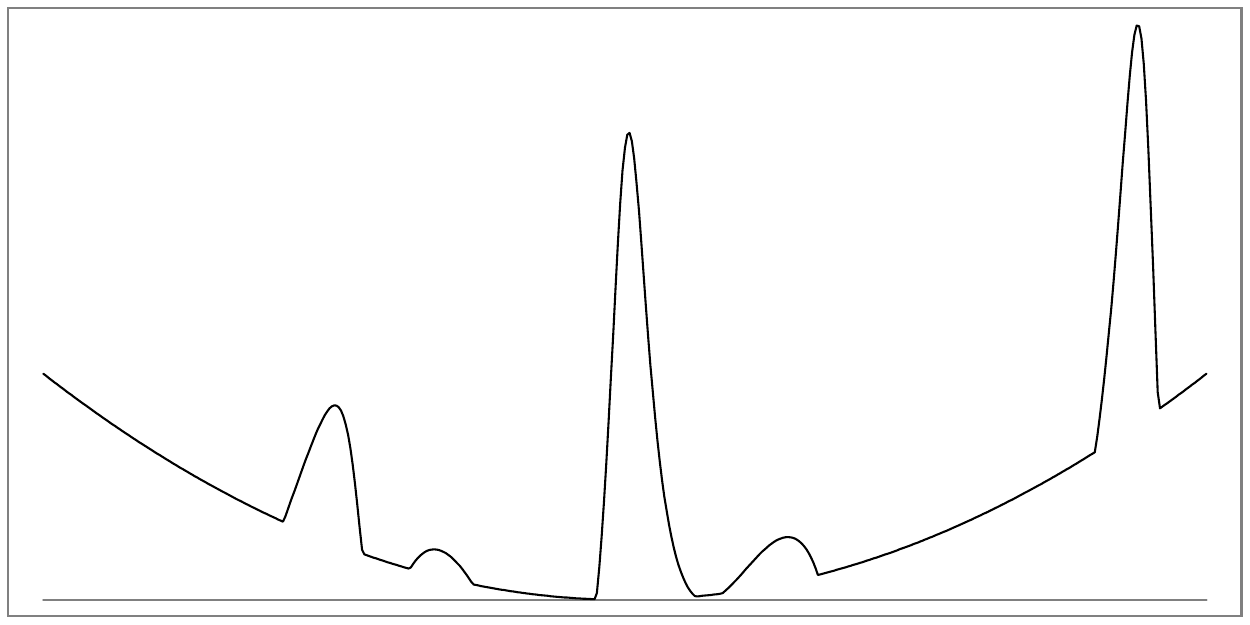}}
\\
\includegraphics[width=.46\hsize]{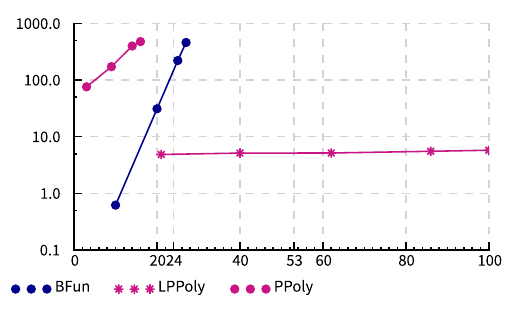}
&
\includegraphics[width=.46\hsize]{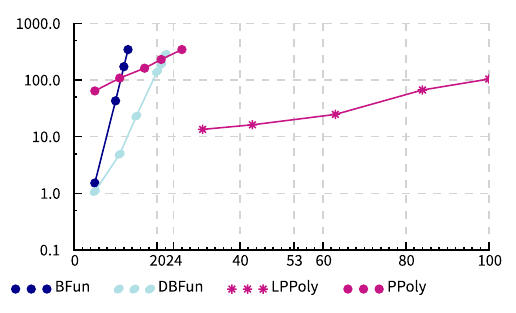}
\\
\bottomrule
\end{tabular}

  \caption{Measurements for functions with division and max}
  \label{fig:exprs-div-max}
\end{center}
\end{figure}

The first two functions in Fig.~\ref{fig:exprs-div-max} are variants
of the Runge family of functions, which have singularities 
in the complex plane near our domain $[-1,1]$.
It is shown in the proof of Theorem \ref{Theorem: bounded division in Gevrey's hierarchy} that the degree of any polynomial approximation to the 
function
$\tfrac{1}{1 + ax^2}$
to error $2^{-n}$
is polynomial in $n$ but exponential in $\log a$.
Thus, these functions are expected to be difficult to approximate by polynomials
even for moderately large values of $a$.  
This turns out to be the case in our implementation,
separating the performance of $\Poly$ from that of $\PPoly$ and $\Frac$.  
Still, $\PPoly$ performs quite poorly for both functions,
which suggests that while our division algorithm runs in polynomial time,
it cannot be considered practically feasible.
However, the local version $\LPPoly$ performs very well on both examples.
The $\Fun$ representations seem to perform with an exponential 
or worse time complexity, which is in line with our complexity results.

The last two functions in Fig.~\ref{fig:exprs-div-max} are
non-smooth and thus cannot be efficiently approximated
by polynomials.  The simpler of these two function is easily handled by the $\Fun$
representations because there is no dependency error, as $x$
appears in the expression effectively only once over each point
in the domain.  As predicted, $\Poly$ cannot cope with these functions
but its local version performs acceptably for the simpler function.
In theory, $\Frac$ should be able to approximate non-smooth functions
as well as $\PPoly$, but we have not yet found an efficient algorithm for this.

Note that $\DBFun$ does better for the last function in Fig.~\ref{fig:exprs-div-max}
than for the very similar function without $\max$.  This again points
to an element of luck due to a high sensitivity of the $\Fun$ representations 
to the quality of the interval extension of $f$.
The local representations have consistently outperformed their global counterparts,
and while the representation $\PPoly$ did quite poorly on some inputs,
its local version performed reasonably well overall. 

\bibliographystyle{abbrv}
\bibliography{fnreps}
\nocite{*}

\appendix
\section{On the uniform complexity of division for functions in Gevrey's hierarchy}
\label{Appendix A}

In this appendix we prove the claim made in the introduction that bounded division
is not polytime computable with respect to the representation for functions in Gevrey's hierarchy that is implicit in \cite{KMRZ}.

An infinitely differentiable function $f \colon [-1,1] \to \R$ belongs to Gevrey's hierarchy
if and only if there exist positive constants $B$, $\ell$ and $\gamma$ such that for all $x \in [-1,1]$ and all $k \in \N$ we have:
\begin{equation}\label{eq: Gev bound}
	\left|f^{(k)}(x)\right| \leq B \ell^k k^{\gamma k}
\end{equation}

The following definition is essentially due to Kawamura, M\"uller, R\"osnick, and Ziegler \cite{KMRZ}.
While the use of explicit representations is avoided throughout \cite{KMRZ}, 
the following is implicit in \cite[Definition 22 (a)]{KMRZ}.

\begin{Definition}\label{Defintion: Gev rep}
The space $\mathcal{G}([-1,1])$ of Gevrey functions on $[-1,1]$ is the represented space of all functions in Gevrey's hierarchy, where a name of a function $f \colon [-1,1] \to \R$ 
is given by a $\Fun$-name of $f$ (see Definition \ref{Definition: representations}) together with positive integer constants 
$B$, $\ell$, and $\gamma$ satisfying \eqref{eq: Gev bound}.
The constant $B$ is encoded in binary, the constant $\ell$ is encoded in unary, and the constant $\gamma$ is given by an encoding of $2^{\gamma}$ in unary.
\end{Definition}

The encodings for the integer constants in Definition \ref{Defintion: Gev rep} are chosen such that a polytime algorithm on the space of Gevrey functions is required to run in polylogarithmic time in $B$, in polynomial time in $\ell$, and in exponential time in $\gamma$. 
This convention ensures that \cite[Theorem 23]{KMRZ} translates to a result on second-order polytime computability on the represented space of Gevrey functions.
One should note that a different representation of the space of Gevrey functions is implicitly given in
\cite[Definition 22 (b)]{KMRZ}, but it is polytime equivalent to the above by virtue of 
\cite[Theorem 23 (a) and (b)]{KMRZ}.

\begin{Theorem}\label{Theorem: bounded division in Gevrey's hierarchy}
Bounded division
\[
	\operatorname{div} \colon \subseteq \mathcal{G}([-1,1]) \times \mathcal{G}([-1,1]) \to \mathcal{G}([-1,1]),
	\;
	(f,g) \mapsto f / g
\]
where
\[
	\dom \operatorname{div} = 
		\Set{(f,g) \in \mathcal{G}([-1,1]) \times \mathcal{G}([-1,1])}
			{g(x) \geq 1 \text{ for all } x \in [-1,1]}
\]
is not polytime computable with respect to the representation given in Definition \ref{Defintion: Gev rep}.
\end{Theorem}
\begin{proof}
Consider the family of polynomial functions:
\[
	f_n(x) = 1 + 2^n x^2.
\]
This sequence is bounded by $1$ from below and uniformly polytime computable with
respect to the above representation of Gevrey functions.

Now consider the sequence of reciprocals:
\[
	g_n(x) = \frac{1}{1 + 2^n x^2}.
\]
If bounded division is polytime computable, then this sequence is again uniformly polytime computable.
We will however show that any sequence of names for $(g_n)_n$ in the above representation grows
super-polynomially.

Let $(B_n)_n$, $(\ell_n)_n$, and $(\gamma_n)_n$ be sequences of natural
numbers satisfying
\[
	\left|g_n^{(k)}(x)\right| \leq B_n \ell_n^k k^{k \gamma_n}
\]
for all $k \in \N$ and all $x \in [-1,1]$.

The function $g_n$ has exactly two singularities in the complex plane:
the imaginary numbers $\tfrac{i}{2^{n/2}}$ and $-\tfrac{i}{2^{n/2}}$.
It follows that the radius of convergence of the Taylor series
of $g_n$ about $0$ is equal to $\tfrac{1}{2^{n/2}}$.
By the Cauchy-Hadamard theorem we obtain for all $n \in \N$:
\[
\limsup_{k \to \infty} \left|\frac{g_n^{k}(0)}{k!}\right|^{1/k} = 2^{n / 2}.
\]
Using the assumption on $(B_n)_n$, $(\ell_n)_n$, and $(\gamma_n)_n$ we obtain for all $n \in \N$ and all $k \in \N$:
\[
	B_n^{1/k} \ell_n \frac{k^{\gamma_n}}{(k!)^{1/k}} \geq 2^{n / 2}.
\]
Use the estimate $k! \geq (k/e)^k$:
\[
	B_n^{1/k} \ell_n e k^{\gamma_n - 1} \geq 2^{n / 2}.
\]
Take binary logarithms on both sides:
\[
	\tfrac{1}{k}\log(B_n) + \log(\ell_n) + \log(e) + (\gamma_n - 1) \log(k) \geq n / 2.
\]
Put $k = 2^{\sqrt{n}}$:
\[
	\tfrac{1}{2^{\sqrt{n}}}\log(B_n) + \log(\ell_n) + \log(e) + \sqrt{n}(\gamma_n - 1) \geq n / 2.
\]
Then at least one of the sequences $(\log(B_n))_n$, $(\ell_n)_n$, or $(2^{\gamma_n})_n$ has to grow at least as fast as 
$2^{\sqrt{n}}$. 
It follows that the size of any sequence of names of $(g_n)_n$ grows super-polynomially.
\end{proof}
\end{document}